\documentclass[reqno]{amsart}
\usepackage{amssymb}
\usepackage{amsmath}
\usepackage{amsfonts}
\usepackage{mathpazo}
\usepackage{hyperref}
\usepackage{multimedia}
\usepackage{graphicx}
\usepackage{acronym}
\usepackage{tikz}

\usetikzlibrary{arrows,decorations.markings}
\usetikzlibrary{arrows,decorations.markings}
\usetikzlibrary{calc}
\newtheorem{theorem}{Theorem}[section]

\newtheorem{lemma}[theorem]{Lemma}
\newtheorem{corollary}[theorem]{Corollary}
\newtheorem{proposition}[theorem]{Proposition}
\newtheorem{definition}[theorem]{Definition}
\newtheorem{example}[theorem]{Example}

\newtheorem{remark}[theorem]{Remark}
\newtheorem{condition}[theorem]{Hypothesis}

\input{tcilatex}
\numberwithin{equation}{section}

\begin{document}
\title[Inverse Scattering Transform]{The Hirota $\tau $-function and
well-posedness of the KdV equation with an arbitrary step like initial
profile decaying on the right half line}
\author{Alexei Rybkin}
\address{University of Alaska Fairbanks}
\date{November, 2010}
\address{Department of Mathematics and Statistics \\
University of Alaska Fairbanks\\
PO Box 756660\\
Fairbanks, AK 99775}
\email{arybkin@alaska.edu}
\thanks{Based on research supported in part by the NSF under grant DMS
0707476.}

\begin{abstract}
We are concerned with the Cauchy problem for the KdV equation on the whole
line with an initial profile $V_{0}$ which is decaying sufficiently fast at $%
+\infty $ and arbitrarily enough (i.e., no decay or pattern of behavior) at $%
-\infty $. We show that this system is completely integrable in a very
strong sense. Namely, the solution $V\left( x,t\right) $ admits the Hirota $%
\tau $-function representation 
\begin{equation}
V\left( x,t\right) =-2\partial _{x}^{2}\log \det \left( I+\mathbb{M}%
_{x,t}\right)   \label{tau}
\end{equation}%
where $\mathbb{M}_{x,t}$ is a Hankel integral operator constucted from
certain scattering and spectral data suitably defined in terms of the
Titchmarsh-Weyl $m$-functions associated with the two half-line Schr\"{o}%
dinger operators corresponding to $V_{0}$. We show that $V\left( x,t\right) $
is real meromorphic with respect to $x$ for any $t>0$. We also show that
under a very mild additional condition on $V_{0}$ representation (\ref{tau})
implies a strong well-posedness of the KdV equation with such $V_{0}$'s.
Among others, our approach yields some relevant results due to Cohen,
Kappeler, Khruslov, Kotlyarov, Venakides, Zhang and others.
\end{abstract}

\maketitle
\tableofcontents


\section{Introduction}

Soliton theory, a major achievement of 20th century mathematics, originated
in 1965 from the fundamental Gardner-Greene-Kruskal-Miura discovery of what
we now call the inverse scattering transform (IST) for the Korteweg-de Vries
(KdV) equation. Conceptually, the IST is similar to the Fourier transform.
In the context of the Cauchy problem for the KdV equation on the full line

\begin{equation}
\partial _{t}V-6V\partial _{x}V+\partial _{x}^{3}V=0  \label{KdV1.1}
\end{equation}%
\begin{equation}
V\left( x,0\right) =V_{0}\left( x\right)  \label{KdVID1.2}
\end{equation}%
the IST method consists, as the standard Fourier transform method, of\ the
following three steps:

\begin{enumerate}
\item the direct transform mapping the initial data $V_{0}(x)$ to a new set
of variables $S_{0}$ in which \eqref{KdV1.1} turns into a very simple first
order linear ordinary equation for $S(t)$ with the initial condition $%
S(0)=S_{0}$;

\item solve then this linear ordinary differential equation for $S(t)$;

\item apply the inverse transform to find $V(x,t)$ from $S(t)$.
\end{enumerate}

In its original edition due to Gardner-Greene-Kruskal-Miura, $S_{0}$ was the
set of scattering data associated with the pair of Schr\"{o}dinger operators 
$(H,H_{0})$ where $H_{0}=-\partial _{x}^{2}$ and $H=H_{0}+V_{0}$. Moreover,
this procedure comes with a beautiful formula

\begin{equation}
V\left( x,t\right) =-2\partial _{x}^{2}\log \det \left( I+\mathbb{M}%
_{x,t}\right) ,  \label{eq1.3}
\end{equation}%
where $\mathbb{M}_{x,t}$ is a two parametric family of integral operators
explicitly constructed in terms of $S(t)$.

Similar methods have also been developed for many other evolution nonlinear
partial differential equations (PDE), which are referred to as completely
integrable.\footnote{%
There is no precise meaning of \textquotedblleft complete integrability" but
the question \textquotedblleft What is integrability?" has drawn much
attention (see e.g. the charming survey \cite{Its2003} by Its).}

Strictly speaking, a complete integrability of \eqref{KdV1.1}-%
\eqref{KdVID1.2} was originally established for $V_{0}$'s rapidly decaying
(short range) at $\pm \infty $ (i.e. in the case of the scattering
theoretical situation for $(H,H_{0})$ ). A suitable analog of IST for the
physically important case of periodic $V_{0}$'s was found around 1974 by
Novikov \cite{Novikov74}. It is based upon the Floquet theory for the
periodic (Hill) Schr\"{o}dinger operator and looks very different from the
case of decaying initial data. There is a conjecture however that these two
cases can actually be unified \cite{AC91}. Some deep related results are
given in \cite{Venak88} and \cite{ErcMcKean90}.

Thus, the KdV equation (as well as other completely integrable by the IST
PDEs) are completely integrable essentially only in these two cases. In
fact, the fundamental question about whether any (physically significant)
well-posed problem (initial value, boundary value, etc.) for equation %
\eqref{KdV1.1} can be solved by a suitable IST, has been raised in one form
or another by many (see e.g. \cite{McLeodOlver83}, \cite{AC91}, \cite%
{KrichNovikov99}). A large amount of effort has been put into developing the
IST for \eqref{KdV1.1} on an interval (finite and semi-infinite). There
seems to be no consensus on whether \eqref{KdV1.1} is indeed completely
integrable in this case but some relevant deep results have been obtained
(see, e.g., \cite{Fokas2002} and the literature therein).

The specific concern of this paper is the problem \eqref{KdV1.1}-%
\eqref{KdVID1.2} with $V_{0}$ outside of classes of rapidly decaying or
periodic functions. Many of known related results are devoted to some sort
of \textquotedblleft hybrids" of the two well-developed cases. Namely, a
physically important case of an initial profile $V_{0}$ in \eqref{KdVID1.2}
which is a short-range perturbation of a step function (a bore type initial
profile) appears to have received the most attention (see e.g. \cite%
{Hruslov76}, \cite{Cohen1984}, \cite{Venak86}, \cite{KK94} to name just a
few and the extensive literature cited therein). Another important case
where IST works is \eqref{KdV1.1}-\eqref{KdVID1.2} with $V_{0}$
representable as a short range perturbation of a half-periodic\footnote{%
I.e. a function which is periodic on a half-line and zero on the other.}
potential (see e.g. \cite{KK94}). A rigorous comprehensive treatment of
steplike short-range perturbations of finite-gap solutions was recently
given in \cite{Teschl2}. We also refer to \cite{Teschl2} for a large account
of literature on initial profiles for which IST works.

Certain cases of slowly decaying profiles have also received considerable
attention (see e.g. \cite{Marchenko91}, \cite{KheNov84}, \cite%
{Gesztesy_Duke92}, \cite{Matveev02}), but this situation is far from being
well-understood.

More literature and discussion of the results obtained therein will be
offered in the main body of the paper. We only mention that in all the
situations above the spectrum of the underlying Schr\"{o}dinger operator $%
-\partial _{x}^{2}+V_{0}$ is much more complicated resulting in new
phenomena (infinite train of solitons, solitons on periodic backgrounds,
singular solutions, etc.). Consequently, much more complicated tools and
harder analysis are required.

We mention that the question of well-posedness\footnote{%
I.e. existence, uniqueness and continuous dependence on the initial data.}
of \eqref{KdV1.1}-\eqref{KdVID1.2} is a serious issue (see, e.g. \cite{Tao06}%
). In the literature on IST this issue is frequently avoided.

The main concern of the present paper is to show that \eqref{KdV1.1}-%
\eqref{KdVID1.2} is completely integrable in a very strong sense for $V_{0}$%
's rapidly decaying at $+\infty $ and essentially arbitrary at $-\infty $.
We will also call such $V_{0}$'s steplike initial profiles. Under very mild
conditions on initial data $V_{0}$ (expressed in terms of the spectrum of
the underlying Schrodinger operator $-\partial _{x}^{2}+V_{0}\left( x\right) 
$), we prove that the representation \eqref{eq1.3} can be extended to our
case. We also establish the analytic smoothing effect and well-posedness of %
\eqref{KdV1.1}-\eqref{KdVID1.2} for initial profiles under consideration.
Emphasize that, as opposed to the relevant previous works, we treat initial
profiles which need not have specific asymptotic behavior at $-\infty $. Our
approach is based upon IST techniques and suitable limiting argument. Note
that the limiting procedures which we employ allow us to recycle many
results from the classical IST method and, in fact, bypass analysis of the
relevant Fredholm integral equation or Riemann-Hilbert problem and work
directly with a Hankel integral operator.

The paper is organized as follows. In Section 2 we merely list some of our
notation and conventions. In Section 3, we review some basics of the full
line and half line Schrodinger operators and the full line short range
scattering theory. In Section 4 we put together some well-known facts on the
KdV equation pertinent to the present paper and specifies what we mean by
the solution to the Cauchy problem for the KdV equation. Section 5 is
devoted to defining a right reflection coefficient from the right incident
for arbitrary potentials. In Section 6 we specify in what sense we
understand Fredholm determinants (Definition \ref{Def Fred}) and prove two
important lemmas. In Section 7 we introduce and study a Hankel integral
operator particularly important in the context of the IST. The main results
(Theorems \ref{Structure} and \ref{MainThm}) are given in Section 8 and
Section 9 is devoted to related discussions and corollaries. We also state
some open problems.


\section{Notation and Preliminaries}

We adhere to standard terminology accepted in Analysis. Namely, $\mathbb{R}%
_{\pm }:=[0,\pm \infty )$, $\mathbb{C}$ is the complex plane, 
\begin{equation*}
\mathbb{C}_{\pm }=\left\{ z\in \mathbb{C}:\pm \func{Im}z>0\right\} .
\end{equation*}%
Unless otherwise stated, we use the subscript $\pm $ to indicate objects
somehow related to $\mathbb{R}_{\pm }$ or $\mathbb{C}_{\pm }$. E.g. $m_{-}$ $%
\left( m_{+}\right) $ denotes the Titchmarsh-Weyl $m$-function (introduced
below) associated with $\mathbb{R}_{-}$ $\left( \mathbb{R}_{+}\right) $ The
bar $\overline{z}$ denotes the complex conjugate of $z$.

We use $\left\Vert \cdot \right\Vert _{X}$ to denote the norm in a Banach
(Hilbert) space $X$. We extensively use Lebesgue spaces $\left( 1\leq
p<\infty ,\ d\mu \text{ is a non-negative measure on a set }S\right) $ 
\begin{align*}
L^{p}\left( S,d\mu \right) & :=\left\{ f:\left\Vert f\right\Vert
_{L^{p}\left( S\right) }:=\left( \int_{S}\left\vert f\left( x\right)
\right\vert ^{p}d\mu \left( x\right) \right) ^{1/p}<\infty \right\} , \\
L^{\infty }\left( S\right) & :=\left\{ f:\left\Vert f\right\Vert _{L^{\infty
}\left( S\right) }:=\limfunc{ess}\sup_{x\in S}\left\vert f\left( x\right)
\right\vert <\infty \right\} , \\
L_{\limfunc{loc}}^{p}\left( S\right) & :=\left\{ \cap L^{p}\left( \Delta
\right) :\Delta \subset S\right\} .
\end{align*}%
and abbreviate them in particular cases as follows ($S$ will typically be $%
\mathbb{R}$ or $\mathbb{R}_{\pm }$):%
\begin{eqnarray*}
L^{p}\left( S,dx\right) =: &&L^{p}\left( S\right) \\
L^{p}\left( \mathbb{R}_{\pm },d\mu \right) &=&:L_{\pm }^{p}\left( d\mu
\right) ,:L^{p}\left( \mathbb{R},d\mu \right) =:L^{p}\left( d\mu \right) \\
L^{p}\left( \mathbb{R}_{\pm }\right) =: &&L_{\pm }^{p},:L^{p}\left( \mathbb{R%
}\right) =:L^{p}
\end{eqnarray*}

Next, $\mathfrak{S}_{2}$ denotes the Hilbert-Schmidt class of linear
operators $A$: 
\begin{equation*}
\mathfrak{S}_2= \left\{A\: : \: \left\Vert A\right\Vert _{\mathfrak{S}_{2}}:=%
\limfunc{tr}\left( A^{\ast }A\right)<\infty \right\}
\end{equation*}%
and $\mathfrak{S}_1$ is the trace class: 
\begin{equation*}
\mathfrak{S}_1= \left\{A\: : \: \left\Vert A\right\Vert _{\mathfrak{S}_{2}}:=%
\limfunc{tr}\left( A^{\ast }A\right) ^{1/2}<\infty \right\}.
\end{equation*}
$\limfunc{Spec}\left( A\right) $ stands for the spectrum of an operator $A$
and $\limfunc{Spec}_{\limfunc{ac}}\left( A\right) $, $\limfunc{Spec}%
_{d}\left( A\right) $ denote the absolutely continuous (a.c.), discrete
components of the (self-adjoint) operator $A$.

Throughout the paper 
\begin{eqnarray*}
\left( \mathcal{F}f\right) \left( \lambda \right) &=&\mathcal{F}f\left(
\lambda \right) =\frac{1}{\sqrt{2\pi }}\int_{\mathbb{R}}e^{i\lambda
x}f\left( x\right) dx \\
\left( \mathcal{F}^{-1}f\right) \left( \lambda \right) &=&\frac{1}{\sqrt{%
2\pi }}\int_{\mathbb{R}}e^{-i\lambda x}f\left( x\right) dx
\end{eqnarray*}%
stand for the Fourier and the inverse Fourier transforms of a tempered
distribution $f$ respectively. If $f\in L^{2}$ and $\limfunc{Supp}f\subseteq 
\mathbb{R}_{\pm }$ then $\left( \mathcal{F}f\right) \left( \lambda \right)
\in H_{\pm }^{2}$. Recall that $H_{\pm }^{p},p>0$, stands for the Hardy
class of analytic on $\mathbb{C}_{\pm }$ functions $f$ such that 
\begin{equation*}
\sup_{\pm y>0}\int_{\mathbb{R}}\left\vert f\left( x+iy\right) \right\vert
^{p}dx<\infty .
\end{equation*}%
Each $H^{p}$- function $f$ admits the estimate (see, e.g. \cite{Garnett})%
\begin{equation}
\left\vert f\left( \lambda \right) \right\vert \lesssim \frac{\left\Vert
f\right\Vert _{H_{+}^{p}}}{\func{Im}\lambda }  \label{Hp}
\end{equation}%
where we have used a convenient convention to write%
\begin{equation*}
x\lesssim y\iff x\leq Cy
\end{equation*}%
with some $C>0$ independent of $x$ and $y$.

Some other miscellaneous notation: $\chi _{S}\left( x\right) $ is the
characteristic function of a set $S$, i.e. 
\begin{equation*}
\chi _{S}\left( x\right) :=\left\{ 
\begin{array}{c}
1,x\in S \\ 
0,x\notin S%
\end{array}%
\right. .
\end{equation*}%
In particular $\chi _{\pm }:=\chi _{_{\mathbb{R}_{\pm }}}$is the Heaviside
function of $\mathbb{R}_{\pm }$.

The following short hand notation will help us keep bulky formulas under
control%
\begin{eqnarray*}
\int_{\mathbb{R}} &=:&\int \\
\left( fg\right) \left( x\right) &:=&f\left( x\right) g\left( x\right)
\end{eqnarray*}%
and

\begin{equation*}
\left\langle x\right\rangle :=\sqrt{1+\left\vert x\right\vert ^{2}}
\end{equation*}%
for an inhomogeneous distance.


\section{The 1D Schr\"{o}dinger operator, scattering and all that}


As its title suggests, in this section we review some basics of the full
line and half line Schrodinger operators and the full line short range
scattering theory.

\subsection{Schr\"{o}dinger operators on the line}

As well known, the inverse scattering transform (IST) method for the KdV
equation is based upon the direct and inverse scattering for the
one-dimensional Schr\"{o}dinger operator. In this section we briefly
introduce the $1D$ $\ $Schr\"{o}dinger operator referring the reader to \cite%
{TeschlBOOK} for precise statements. Through the paper 
\begin{equation}
H_{0}=-\partial _{x}^{2}  \label{free}
\end{equation}%
is the free (unperturbed) Schr\"{o}dinger operator on the Hilbert space $%
L^{2}$. Given a real locally integrable function $V$, called a potential,
define on $L^{2}$ the (perturbed) Schr\"{o}dinger operator:
\begin{equation}
H=H_{0}+V\left( x\right) =-\partial _{x}^{2}+V\left( x\right)  \label{Schr}
\end{equation}%
and two half-line Schr\"{o}dinger operators with the Dirichlet boundary
condition at $x=0$%
\begin{equation}
H_{\pm }^{D}=-\partial _{x}^{2}+V\left( x\right) \text{ on }L_{\pm }^{2}%
\text{ with }u\left( \pm 0\right) =0.  \label{Diri}
\end{equation}%
We shall assume that each Schr\"{o}dinger operator considered in this paper
is self-adjoint on its natural domain in $L^{2}$. This assumption
automatically puts a certain restriction on $V$'s:


\begin{condition}
\label{H1}

\begin{enumerate}
\item Reality%
\begin{equation*}
\overline{V}\left( x\right) =V\left( x\right)
\end{equation*}

\item Local integrability%
\begin{equation*}
V\in L_{loc}^{1}
\end{equation*}

\item $V$ is limit point case at $\pm \infty $%
\begin{equation*}
V\in \text{l.p.}\left( \pm \infty \right)
\end{equation*}
\end{enumerate}
\end{condition}

Hypothesis \ref{H1} means that the minimal operator generated by the
differential expression $-\partial _{x}^{2}+V\left( x\right) $ in the space $%
L_{\pm }^{2}$ has deficiency indices $\left( 1,1\right) $ as opposed to the
limit circle case when the deficiency indices are $\left( 2,2\right) $.
While no explicit description (e.g. in terms of potentials) of the limit
point/circle classification is currently available but it is well-known that
the class $V\in $ l.p.$\left( \pm \infty \right) $ is extremely broad. In
particular, all physically meaningful initial profiles (no decay of any kind
is assumed) in the KdV equation are limit point at $\pm \infty $. \ We will
refer to potentials subject to Hypothesis \ref{H1} as arbitrary.


\subsection{The Titchmarsh-Weyl $m$-function}

The material of this subsection is classical and standard (see, e.g. \cite%
{TeschlBOOK}). Assuming Hypothesis \ref{H1}, consider 
\begin{equation*}
-\partial _{x}^{2}u+V\left( x\right) u=zu,x\in \mathbb{R}_{\pm }.
\end{equation*}

If $V\in $l.p.$\left( \pm \infty \right) $ then there exists a unique (up to
a multiplicative constant) solution, called Weyl, such that $\Psi _{\pm
}\left( x,z\right) \in L_{\pm }^{2}$ for each $z\in \mathbb{C}_{+}$.

\begin{definition}
\label{def1} The\ function%
\begin{equation*}
m_{\pm }\left( z\right) =\frac{\partial _{x}\Psi _{\pm }\left( 0,z\right) }{%
\Psi _{\pm }\left( 0,z\right) },\ \ z\in \mathbb{C}_{+}
\end{equation*}%
is called (Dirichlet, principal) Titchmarsh-Weyl $m-$function.
\end{definition}

Alternatively, the $m$-function can be defined 
\begin{equation*}
m_{\pm }\left( z\right) =\lim_{0<x<y\rightarrow 0}\partial _{xy}^{2}G_{\pm
}\left( x,y;z\right)
\end{equation*}%
with $G_{\pm }\left( x,y;z\right) $ the Green's function of the
corresponding half line Schr\"{o}dinger operator. Properties of the $m$%
-function include

\begin{enumerate}
\item The Herglotz property: $m:\mathbb{C}_{+}\rightarrow \mathbb{C}_{+}$
and analytic

\item The Herglotz representation theorem: there is a non-negative measure $%
\mu $ subject to $\int \left\langle \lambda \right\rangle ^{-2}d\mu \left(
\lambda \right) $ such that%
\begin{equation}
m\left( z\right) =\func{Re}m\left( i\right) +\int_{\mathbb{R}}\dfrac{t+z}{t-z%
}\frac{d\mu }{1+t^{2}}.  \label{Herglotz3.3}
\end{equation}%
The measure $\mu $ is the spectral measure of $H_{\pm }^{D}$ introduced by (%
\ref{Diri}) and can be computed by the Herglotz inversion formula\footnote{%
Through the paper we use the convention%
\begin{equation*}
\func{Im}f\left( t+i0\right) dt:=w-\lim_{\varepsilon \rightarrow +0}\func{Im}%
f\left( t+i\varepsilon \right) dt.
\end{equation*}%
} 
\begin{equation*}
d\mu =\frac{1}{\pi }\func{Im}m\left( t+i0\right) dt.
\end{equation*}

\item The Borg-Marchenko Uniqueness Theorem:%
\begin{equation*}
m_{1}=m_{2}\quad \Longrightarrow \quad V_{1}=V_{2}.
\end{equation*}
\end{enumerate}


\subsection{Scattering theory for the Schr\"{o}dinger operator on the line}

To fix our notation and terminology we give a brief introduction to $1D$
scattering theory for the Schr\"{o}dinger operator. We refer the reader to
the classical paper \cite{Deift79} (where the notation slightly differs from
ours though). Through this section we deal with a pair $\left(
H,H_{0}\right) $ of Schr\"{o}dinger operators $\left( \ref{Schr}\right) $
and $\left( \ref{free}\right) $ with a Faddeev potential $V$. I.e. 
\begin{equation*}
V\in L^{1}\left( \left\langle x\right\rangle dx\right) =\left\{ f:\int
\left( 1+\left\vert x\right\vert \right) \left\vert f\left( x\right)
\right\vert dx<\infty \right\} .
\end{equation*}%
Under this assumption on $V$ one has a typical scattering theoretical
situation which means that all four wave operators and the scattering
operator for the pair $\left( H,H_{0}\right) $ exist. In particular, the
absolutely continuous (a.c.) part of $H$ is unitary equivalent to $H_{0}$.
For $\limfunc{Spec}\left( H\right) $ we have:%
\begin{equation*}
\func{Spec}\left( H\right) =\func{Spec}_{d}\left( H\right) \cup \func{Spec}%
_{ac}\left( H\right) ,
\end{equation*}%
where the discrete spectrum $\limfunc{Spec}_{d}\left( H\right) =\{-\kappa
_{n}^{2}\}_{n=1}^{N}$ is negative, simple and 
\begin{equation*}
N\leq 1+\left\Vert V\right\Vert _{L^{1}\left( \left\langle x\right\rangle
dx\right) }
\end{equation*}%
and the absolutely continuous (a.c.) spectrum $\limfunc{Spec}_{ac}\left(
H\right) =\mathbb{R}_{+}$ and of multiplicity two (with no embedded
eigenvalues).

Since the a.c. spectrum of $H$ is of uniform multiplicity two, the
scattering matrix $S$ (the scattering operator in the spectral
representation of $H_{0}$) is a two by two unitary matrix

\begin{equation*}
S\left( \lambda \right) =\left( 
\begin{array}{cc}
T\left( \lambda \right) & R\left( \lambda \right) \\ 
L\left( \lambda \right) & T\left( \lambda \right)%
\end{array}%
\right) ,\lambda ^{2}\in \func{Spec}_{\limfunc{ac}}\left( H\right) =\mathbb{R%
}_{+},
\end{equation*}%
where $T$, $L$ and $R$ denote respectively the transmission, reflection
coefficients from the left and right incident. Due to unitarity of $S$ one
has (for a.e. $\lambda \in \mathbb{R}$) 
\begin{equation}
\left\vert T\left( \lambda \right) \right\vert ^{2}+\left\vert R\left(
\lambda \right) \right\vert ^{2}=1,\:\left\vert T\left( \lambda \right)
\right\vert ^{2}+\left\vert L\left( \lambda \right) \right\vert ^{2}=1
\label{2-3}
\end{equation}%
\begin{equation}
\overline{T\left( \lambda \right) }R\left( \lambda \right) +T\left( \lambda
\right) \overline{L\left( \lambda \right) }=0  \label{2'}
\end{equation}%
\begin{equation*}
T\left( -\lambda \right) =\overline{T\left( \lambda \right) },\: R\left(
-\lambda \right) =\overline{R\left( \lambda \right) },\: L\left( -\lambda
\right) =\overline{L\left( \lambda \right) }.
\end{equation*}%
The quantities $T$, $L$ and $R$ are related to the existence of special
solutions $\psi _{\pm }$ (called Jost or Faddeev) to the (stationary) Schr%
\"{o}dinger equation 
\begin{equation*}
-\partial _{x}^{2}u+V\left( x\right) u=\lambda ^{2}u,~\lambda \in \mathbb{R},
\end{equation*}%
asymptotically behaving as

\begin{equation*}
\psi _{+}\left( x,\lambda \right) \sim \left\{ 
\begin{array}{c}
T\left( \lambda \right) e^{i\lambda x},\ \ \ \ \ \ \ \ \ \ \ x\rightarrow
\infty , \\ 
e^{i\lambda x}+L\left( \lambda \right) e^{-i\lambda x},x\rightarrow -\infty ,%
\end{array}%
\right.
\end{equation*}%
\begin{equation}
\psi _{-}\left( x,\lambda \right) \sim \left\{ 
\begin{array}{c}
e^{-i\lambda x}+R\left( \lambda \right) e^{i\lambda x},x\rightarrow \infty ,
\\ 
T\left( \lambda \right) e^{-i\lambda x},\ \ \ \ \ \ \ \ \ \ \ \ x\rightarrow
-\infty .%
\end{array}%
\right. .  \label{ksi-}
\end{equation}%
Scattering solutions $\psi _{\pm }\left( x,\lambda \right) $ can be obtained
by solving the Volterra integral equations\footnote{%
In the literature, $y_{\pm }$ are typically denoted by $m_{1,2}$ and called
Faddeev or Jost functions. In our exposition the letter $m$ is reserved for
the $m$-function.} 
\begin{equation}
y_{\pm }\left( x,\lambda \right) =1\pm \int_{x}^{\pm \infty }\frac{%
e^{2i\lambda \left\vert s-x\right\vert }-1}{2i\lambda }V\left( s\right)
y_{\pm }\left( s,\lambda \right) ds  \label{y+-}
\end{equation}%
for%
\begin{equation*}
y_{\pm }\left( x,\lambda \right) :=\frac{\psi _{\pm }\left( x,\lambda
\right) }{T(\lambda )}e^{\mp i\lambda x}.
\end{equation*}%
For the transition coefficients $T,\ R,\ L$ one has%
\begin{equation}
T\left( \lambda \right) =\left( 1-\frac{1}{2i\lambda }\int V\left( x\right)
y_{+}\left( x,\lambda \right) dx\right) ^{-1}  \label{T}
\end{equation}%
\begin{equation}
R\left( \lambda \right) =\frac{T\left( \lambda \right) }{2i\lambda }\int
e^{-2i\lambda x}V\left( x\right) y_{-}\left( x,\lambda \right) dx  \label{R1}
\end{equation}%
\begin{equation}
L\left( \lambda \right) =\frac{T\left( \lambda \right) }{2i\lambda }\int
e^{2i\lambda x}V\left( x\right) y_{+}\left( x,\lambda \right) dx.  \label{L}
\end{equation}


\begin{remark}
The assumption $V\in L^{1}\left( \left\langle x\right\rangle dx\right) $ can
be relaxed to $V\in L^{1}$ in most of statements of this section. The number
of negative eigenvalues (bound states) $\left\{ -\varkappa _{n}^{2}\right\}
_{n\geq 1}$ may become infinite accumulating to $0$. The latter implies more
complicated behavior of the scattering matrix at $\lambda =0$.
\end{remark}

The following facts from \cite{Deift79} will be important:

\begin{enumerate}
\item If $V\in L^{1}\left( \left\langle x\right\rangle ^{2}dx\right) $ and
generic, i.e. 
\begin{equation}
T\left( \lambda \right) =\alpha \lambda +o\left( \lambda \right) ,\lambda
\rightarrow 0,\alpha \neq 0,  \label{generic}
\end{equation}%
then the function

\begin{eqnarray}
f\left( \lambda \right) &:=&\sqrt{2\pi }\frac{T\left( \lambda /2\right) }{%
i\lambda }=\sqrt{2\pi }\left( i\lambda -\int V\left( x\right) y_{+}\left(
x,\lambda /2\right) dx\right) ^{-1}  \label{f} \\
&=&O\left( 1/\lambda \right) ,\lambda \rightarrow \infty ,  \notag
\end{eqnarray}%
and is continuous on $\mathbb{R}$. A potential $V\in L^{1}\left(
\left\langle x\right\rangle ^{2}dx\right) $ for which $\left( \ref{generic}%
\right) $ doesn't hold is called exceptional and it can be turned into a
generic one by an arbitrary small deformation.

\item There exists a function $g$ subject to%
\begin{equation}
\left\vert g\left( x\right) \right\vert \leq \left\vert V\left( x\right)
\right\vert +\limfunc{const}W\left( x\right) ,\ \ W\left( x\right) :=\left\{ 
\begin{array}{c}
\int_{x}^{\infty }\left\vert V\right\vert ,x\geq 0 \\ 
\int_{-\infty }^{x}\left\vert V\right\vert ,x<0%
\end{array}%
\right.  \label{cond on g}
\end{equation}%
such that 
\begin{equation}
R\left( \lambda /2\right) =f\left( \lambda \right) \left( \mathcal{F}%
^{-1}g\right) \left( \lambda \right)  \label{R(lambda/2)}
\end{equation}%
where $f$ is defined by $\left( \ref{f}\right) $.

\item If $V\in L^{1}\left( \left\langle x\right\rangle ^{2}dx\right) $ then 
\begin{equation}
\left\Vert \partial _{\lambda }y_{+}\left( x,\cdot \right) \right\Vert
_{L^{\infty }}\lesssim \left\langle x\right\rangle ^{2}.  \label{fact 3}
\end{equation}
\end{enumerate}

We will need the following technical lemma.


\begin{lemma}
\label{Lemma 1}If $V\in L^{2}\left( \left\langle x\right\rangle
^{2}dx\right) $ and generic then $\partial _{\lambda }R\in L^{2}.$
\end{lemma}

\begin{proof}
From $\left( \ref{R(lambda/2)}\right) $%
\begin{equation*}
\partial _{\lambda }R\left( \lambda /2\right) =\partial _{\lambda }f\left(
\lambda \right) \left( \mathcal{F}^{-1}g\right) \left( \lambda \right)
+f\left( \lambda \right) \partial _{\lambda }\left( \mathcal{F}^{-1}g\right)
\left( \lambda \right)
\end{equation*}%
and hence 
\begin{eqnarray*}
\left\Vert \partial _{\lambda }R\right\Vert _{L^{2}} &\lesssim &\left\Vert
\partial _{\lambda }f\right\Vert _{L^{\infty }}\left\Vert \mathcal{F}%
^{-1}g\right\Vert _{L^{2}}+\left\Vert f\right\Vert _{L^{\infty }}\left\Vert 
\mathcal{F}^{-1}xg\right\Vert _{L^{2}} \\
&=&\left\Vert \partial _{\lambda }f\right\Vert _{L^{\infty }}\left\Vert
g\right\Vert _{L^{2}}+\left\Vert f\right\Vert _{L^{\infty }}\left\Vert
xg\right\Vert _{L^{2}}.
\end{eqnarray*}%
Due to $\left( \ref{f}\right) $ and $\left( \ref{cond on g}\right) $ $%
\left\Vert f\right\Vert _{L^{\infty }}$ and $\left\Vert g\right\Vert
_{L^{2}} $ are both finite. The condition $V\in L^{2}\left( \left\langle
x\right\rangle ^{2}dx\right) $ implies by $\left( \ref{cond on g}\right) $
that $\left\Vert xg\right\Vert _{L^{2}}\lesssim \left\Vert g\right\Vert
_{L^{2}\left( \left\langle x\right\rangle ^{2}dx\right) }$ is finite.
Indeed, setting $W_{+}\left( x\right) :=\chi _{+}\left( x\right)
\int_{x}^{\infty }\left\vert V\right\vert $ and $W_{-}\left( x\right) :=\chi
_{-}\left( x\right) \int_{-\infty }^{x}\left\vert V\right\vert $ one has 
\begin{eqnarray*}
\left\Vert W_{\pm }\right\Vert _{L_{\pm }^{2}}^{2} &=&\int_{0}^{\infty
}\left( \int_{x}^{\infty }\left\vert V\left( \pm s\right) \right\vert
ds\right) ^{2}dx=2\int_{0}^{\infty }x\left\vert V\left( \pm x\right)
\right\vert \int_{x}^{\infty }\left\vert V\left( \pm s\right) \right\vert
dsdx \\
&\leq &2\left\Vert xV\right\Vert _{L_{\pm }^{2}}\left\Vert W_{\pm
}\right\Vert _{L_{\pm }^{2}}
\end{eqnarray*}%
and hence%
\begin{equation*}
\left\Vert W_{\pm }\right\Vert _{L_{\pm }^{2}}\leq 2\left\Vert xV\right\Vert
_{L_{\pm }^{2}}\leq 2\left\Vert V\right\Vert _{L_{\pm }^{2}\left(
\left\langle x\right\rangle ^{2}dx\right) }.
\end{equation*}%
We therefore have $g\in L^{2}\left( \left\langle x\right\rangle
^{2}dx\right) $ and hence $\left\Vert xg\right\Vert _{L^{2}}$ is finite. It
only remains to show that $\partial _{\lambda }f$ is bounded. By direct
differentiation of $\left( \ref{f}\right) $ one has 
\begin{equation*}
\partial _{\lambda }f\left( \lambda \right) =\frac{-1}{\sqrt{2\pi }}%
f^{2}\left( \lambda \right) \left( i-\int V\left( x\right) \partial
_{\lambda }y_{+}\left( x,\lambda /2\right) dx\right)
\end{equation*}%
and by $\left( \ref{fact 3}\right) $%
\begin{equation*}
\left\vert \partial _{\lambda }f\right\vert \lesssim \left\vert f\right\vert
^{2}\left( 1+C\left\Vert V\right\Vert _{L^{1}\left( \left\langle
x\right\rangle ^{2}dx\right) }\right)
\end{equation*}%
and due to $\left( \ref{f}\right) $ the statement is proven.
\end{proof}


\section{The KdV, IST and all that}


In this section we merely put together some well-known information on the
KdV equation pertinent to the present paper. Definition \ref{def2} below
specifies what we mean by the solution to the Cauchy problem for the KdV
equation.

\subsection{The classical Inverse Scattering Transform (IST)}

To specify our notation and state pivotal equations we briefly outline basic
ideas of the IST methods originated from the seminal 1965 work of
Gardner-Greene-Kruskal-Miura (see, e.g. \cite{AC91} and the very extensive
literature cited therein). In the context of the initial value problem for
the KdV equation on $\mathbb{R}$%
\begin{equation}
\partial _{t}V-6V\partial _{x}V+\partial _{x}^{3}V=0  \label{KdV}
\end{equation}%
\begin{equation}
V\left( x,0\right) =V_{0}\left( x\right)  \label{KdVID}
\end{equation}%
where $V(x,t)$ is subject to\footnote{%
such solutions are referred to as rapidly decaying. For simplicity we call
them short range.} ($l=0,1,2,3$) 
\begin{equation*}
\sup_{t\geq 0}\left\Vert \partial _{x}^{l}V\left( x,t\right) \right\Vert
_{L^{1}\left( \left\langle x\right\rangle dx\right) }<\infty
\end{equation*}%
the classical inverse scattering formalism for $\left( \ref{KdV}\right)
-\left( \ref{KdVID}\right) $ goes as follows. Associate with $\left( \ref%
{KdV}\right) $ a one parametric family of full line Schr\"{o}dinger
operators $H\left( t\right) =-\partial _{x}^{2}+V(x,t)$. Form now the
scattering data\emph{\ }%
\begin{equation}
\mathcal{S}\left( t\right) :=\{R(\lambda ,t),\{-\kappa _{n}^{2}\left(
t\right) ,c_{n}\left( t\right) \}\}  \label{2.17}
\end{equation}%
where $\left\{ c_{n}\left( t\right) \right\} $ are the norming constants
corresponding to bound states\ $\{-\kappa _{n}^{2}\left( t\right) \}$. It is
well-known that the map $V(x,t)\rightarrow \mathcal{S}\left( t\right) $ is
one-to-one.

The fundamental fact of inverse scattering formalism is that the map (time
evolution) $t\rightarrow \mathcal{S}\left( t\right) $ has a very simple form 
\begin{equation}
R(\lambda ;t)=R(\lambda )e^{8ik^{3}t},\kappa _{n}(t)=\kappa _{n},c_{n}\left(
t\right) =c_{n}e^{-4\kappa _{n}^{3}t}.  \label{2.18}
\end{equation}

The problem $\left( \ref{KdV}\right) -\left( \ref{KdVID}\right) $ can now be
solved in three steps. Solve the direct scattering problem $V_{0}\left(
x\right) \rightarrow \mathcal{S}\left( 0\right) .$\ Find next the time
evolution $\mathcal{S}\left( t\right) $ by $\left( \ref{2.17}\right) $ and $%
\left( \ref{2.18}\right) $ and finally solve the inverse scattering problem $%
\mathcal{S}\left( t\right) \rightarrow V(x,t)$. The last step can be done by
any applicable method. For instance, one can solve $\mathcal{S}\left(
t\right) \rightarrow V(x,t)$ as a Riemann-Hilbert problem%
\begin{equation*}
\psi _{+}\left( x,-\lambda ,t\right) +R(\lambda )e^{8ik^{3}t}\psi _{+}\left(
x,\lambda ,t\right) =T\left( \lambda \right) \psi _{-}\left( x,\lambda
,t\right) ,\: \lambda \in \mathbb{R}
\end{equation*}%
for $\psi _{\pm }$. The function 
\begin{equation*}
V(x,t)=\frac{\partial _{x}^{2}\psi _{\pm }\left( x,\lambda ,t\right) }{\psi
_{\pm }\left( x,\lambda ,t\right) }+\lambda ^{2}
\end{equation*}%
then solves $\left( \ref{KdV}\right) $, the procedure being independent of
the choice of $\pm $. Alternatively, one can solve the inverse problem $%
\mathcal{S}\left( t\right) \rightarrow V(x,t)$ by means of the Marchenko
procedure\footnote{%
This procedure is also referred to as the Gelfand-Levitan-Marchenko.} which
essentially boils down to the nice formula 
\begin{equation*}
V\left( x,t\right) =-2\partial _{x}^{2}\log \det \left( I+\mathbb{M}%
_{x,t}\right)
\end{equation*}%
where $\mathbb{M}_{x,t}:L_{+}^{2}\rightarrow L_{+}^{2}$ is a two parametric
family of integral operators%
\begin{equation}
\left( \mathbb{M}_{x,t}f\right) \left( y\right) =\int_{0}^{\infty
}M_{x,t}\left( y+s\right) f\left( s\right) ds,\ \ \ f\in L_{+}^{2},
\label{ClassM'}
\end{equation}%
with the kernel 
\begin{align}
M_{x,t}\left( \cdot \right) &:=M\left( \cdot +2x,t\right) ,  \label{classM}
\\
M\left( y,t\right) &:=\dsum_{n=1}^{N}c_{n}^{2}e^{8\kappa
_{n}^{3}t}e^{-\kappa _{n}y}+\frac{1}{2\pi }\int e^{i\lambda y}R(\lambda
)e^{8i\lambda^{3}t}d\lambda .  \notag
\end{align}


\begin{definition}
The operator $\mathbb{M}_{x,t}$ defined by $\left( \ref{ClassM'}\right)
,\left( \ref{classM}\right) $ is called a (time evolved) Marchenko operator%
\footnote{%
Also referred to as Gelfand-Levitan, Gelfand-Levintan-Marchenko or
Faddeev-Marchenko.
\par
{}} associated with the scattering data $\left( \ref{2.17}\right) $ and $%
\left( \ref{2.18}\right) $.
\end{definition}

We have actually considered the left Marchenko operator. The right Marchenko
operator can be introduced in a similar manner but will not admit a proper
generalization to our setting.

Over the last forty years soliton theory has experienced a rapid development
through efforts by the math, science and engineering communities and the
literature on the subject is enormously extensive and diverse. Some
literature relevant to our consideration have already been given in
Introduction and some more will be given below. In addition to this we
mention here that certain IST type schemes are also available for the
so-called finite gap algebro-geometric solutions to KdV. (see e.g. \cite%
{GH03} where extensive updated literature is given.) The IST methods are
quite different in this context and based on analysis on Riemann surfaces.

A comprehensive account of classes of initial data for which the IST is
rigorously developed is given in the recent paper \cite{Teschl2}.


\subsection{Well-posedness (WP) of the KdV equation\label{WP}}

The interest in WP problems arose almost at the same time as the IST boom
started but they are typically approached by means of PDEs techniques \cite%
{Tao06} (norm estimates, etc.) and the IST is not usually employed. The
opposite is quite typical instead: assuming WP one applies the IST method to
find the unique solution to KdV. There is also a considerable gap between
classes of $V_{0}$'s for which WP is established and those $V_{0}$'s for
which the IST is rigorously justified, the former being much wider than the
latter.

Solutions of the KdV can be understood in a number of different senses \cite%
{Tao06} (classical, strong, weak, etc.) resulting in a variety of different
well-posedness results. WP issues are not in the focus of the present paper
and we do not attempt to give a comprehensive survey. We mention only the
recent sharp results on global well-posedness in $H^{-3/4}\left( \mathbb{R}%
\right) $ \cite{Guo09} (which extends \cite{ColKeStaTao03} where it was
proven for\footnote{%
We recall $f\in H^{s},s\in \mathbb{R}$, if $\mathcal{F}f\in L^{2}\left(
\langle \lambda \rangle ^{s}d\lambda \right) $.} $H^{-s}\left( \mathbb{R}%
\right) $ with $s<3/4$) and a similar result, in the periodic context, \cite%
{KapTop06} where well-posedness is proven in $H^{-1}\left( \mathbb{T}\right) 
$. Note that the approach of \cite{KapTop06} utilizes complete integrability
in a crucial way.

We understand WP in a strong way.


\begin{definition}
\label{def2} Let $\left\{ V_{n}\left( x,t\right) \right\} ,\ x\in \mathbb{R}$
and $t\geq 0$ be a sequence of classical solutions of (\ref{KdV}) with the
compactly supported initial data%
\begin{equation*}
V_{n}\left( x,0\right) =V_{0,n}\left( x\right)
\end{equation*}%
approximating $V_{0}\left( x\right) $ in $L_{\limfunc{loc}}^{2}$. We call $%
V\left( x,t\right) ,\ x\in \mathbb{R}$ and $t\geq 0$, a global natural
solution to (\ref{KdV}), (\ref{KdVID}) if $V$ is a classical solution and%
\begin{equation*}
V\left( x,t\right) =\lim_{n\rightarrow \infty }V_{n}\left( x,t\right)
\end{equation*}%
uniformly on any compact for any $t>0$ independently of the choice of $%
\left\{ V_{0,n}\right\} $.
\end{definition}

Our choice of definition is motivated by the methods we employ and it also
looks quite natural from the computational and physical point of view.


\section{The reflection coefficient}

In this section we define a right reflection coefficient $R$ from the right
incident for arbitrary potentials (i.e. subject to Hypothesis \ref{H1}) on $%
\mathbb{R}_{-}$ and $L^{1}$ on $\mathbb{R}_{+}$. It is convenient to
fragment 
\begin{equation}
V\left( x\right) =V_{-}\left( x\right) +V_{+}\left( x\right) \text{ }
\label{Vfragm}
\end{equation}%
into two potentials $V_{\pm }\left( x\right) =V\left( x\right) \chi _{\pm
}\left( x\right) $ supported on $\mathbb{R}_{\pm }$ and consider the
reflection coefficient from $V_{\pm }$ separately and then combine them.
Recall that we have agreed to write $f_{\pm }$ ($f$ could be an operator,
space, scattering quantity, $m$-function, etc.) with $\pm $ if it is
associated with $\mathbb{R}_{\pm }$. 


\subsection{Potentials supported on a half line}

Assume first that $V=V_{-}$ is supported on $\mathbb{R}_{-}$ and subject to
Hypothesis \ref{H1}. The Schr\"{o}dinger equation then has a solution $\Psi
\left( x,\lambda \right) $ such that for any real $\lambda $ 
\begin{equation*}
\Psi \left( x,\lambda \right) =\left\{ 
\begin{array}{ccc}
C\left( \lambda \right) \Psi _{-}\left( x,\lambda \right) & , & x<0 \\ 
e^{-i\lambda x}+R_{-}\left( \lambda \right) e^{i\lambda x} & , & x\geq 0%
\end{array}%
\right.
\end{equation*}%
where $\Psi _{-}$ is the Weyl solution and $C$ and $R_{-}$ are some
coefficients\footnote{%
We remind that in our notation $R_{-}$ stands for the right reflection
coefficient off the potential $V_{-}$. In the literature $R_{-}\,$\ also
denotes the left reflection coefficient.}. Note that $\Psi _{-}$ turns into
the Jost solution $\psi _{-}$ $\left( \ref{ksi-}\right) $ if $V$ is from the
Faddeev class. The continuity of $\Psi \left( x,\lambda \right) $ and its
derivative at $x=0$ immediately implies that for a.e. real $\lambda $ 
\begin{equation*}
R_{-}\left( \lambda \right) =\frac{i\lambda -m_{-}\left( \lambda
^{2}+i0\right) }{i\lambda +m_{-}\left( \lambda ^{2}+i0\right) }.
\end{equation*}%
Due to the analyticity of $m_{-}$ and the symmetry property 
\begin{equation*}
m_{-}\left( \overline{z}\right) =\overline{m_{-}\left( z\right) }
\end{equation*}%
the function $R_{-}\left( \lambda \right) $ can be analytically continued
into the upper half plane and%
\begin{equation*}
R_{-}\left( -\overline{\lambda }\right) =\overline{R_{-}\left( \lambda
\right) }
\end{equation*}%
for any $\lambda \in \mathbb{C}_{+}$ except for those purely imaginary $%
\lambda $'s for which $\lambda ^{2}\in \limfunc{Spec}\left( H_{-}\right) $.
For real $\lambda $'s one can easily see that 
\begin{equation*}
\left\vert R_{-}\left( \lambda \right) \right\vert \leq 1.
\end{equation*}%
The next important property of $R$ is related to inverse problems. By the
Borg-Marchenko uniqueness, $m_{-}$ determines $V$ and hence $R$ also
determines $V$. Due to analyticity this means that the knowledge of $R\left(
\lambda \right) $ on any set of real $\lambda $'s of positive Lebesgue
measure determines $V\left( x\right) $ for a.e. $x<0$. Therefore, no
additional information about bound states and their norming constants $%
\{-\kappa _{n}^{2},c_{n}\}$ is required in our case. In fact $\{i\kappa
_{n}\}$ are the (simple) poles of $R$ in the upper half plane with residues $%
\limfunc{Res}\left( R,i\kappa _{n}\right) =ic_{n}^{2}$ where $c_{n}$ are
norming constants \cite{AK01}. For the reader's convenience we summarize
what we have said as


\begin{proposition}
\label{Prop_R}Let $H_{-}=-\partial _{x}^{2}+V_{-}\left( x\right) $ be the
Schr\"{o}dinger operator on $L^{2}$ with $V_{-}$ supported on $\mathbb{R}%
_{-} $ subject to Hypothesis \ref{H1}. Let $m_{-}$ be the Dirichlet
Titchmarsh-Weyl $m$-function of $-\partial _{x}^{2}+V_{-}\left( x\right) $
corresponding to $\mathbb{R}_{-}$. Then the right reflection coefficient $%
R_{-}\left( \lambda \right) $ is given by%
\begin{equation}
R_{-}\left( \lambda \right) =\frac{i\lambda -m_{-}\left( \lambda ^{2}\right) 
}{i\lambda +m_{-}\left( \lambda ^{2}\right) }=-1+\frac{2i\lambda }{i\lambda
+m_{-}\left( \lambda ^{2}\right) }  \label{R}
\end{equation}%
and it represents an analytic in the upper half plane function except for
those $\lambda $'s on the imaginary line for which $\lambda ^{2}\in \limfunc{%
Spec}\left( H_{-}\right) $. \ Furthermore, it is symmetric with respect to
the imaginary axis, i.e.%
\begin{equation*}
R_{-}\left( -\overline{\lambda }\right) =\overline{R_{-}\left( \lambda
\right) }
\end{equation*}%
and contractive on the real line: 
\begin{eqnarray*}
\left\vert R_{-}\left( \lambda \right) \right\vert &\leq &1\text{ \ \ \ for
a.e.\ \ \ }\lambda \in \mathbb{R} \\
\left\vert R_{-}\left( \lambda \right) \right\vert &<&1\text{ \ \ \ for
a.e.\ \ \ }\lambda \in \func{Spec}_{ac}\left( H_{-}\right) .
\end{eqnarray*}%
The function $R$ may have simple poles $\{i\kappa _{n}\}$ on the positive
part of the imaginary axis. Moreover, the set $\{-\kappa _{n}^{2}\}$
coincides with the negative discrete spectrum of $H_{-}$ and 
\begin{equation*}
\limfunc{Res}\left( R,i\kappa _{n}\right) =ic_{n}^{2},
\end{equation*}%
where $c_{n}$ is the norming constant corresponding to the bound state $%
-\kappa _{n}^{2}$. If $V_{-}$ is short range then $R_{-}$ defined by $\left( %
\ref{R}\right) $ and $\left( \ref{R1}\right) $ agree.
\end{proposition}

The statement for the left reflection coefficient $L_{+}$ associated with $%
H_{+}$ is almost identical to Proposition \ref{Prop_R} with 
\begin{equation}
L_{+}\left( \lambda \right) =\frac{i\lambda -m_{+}\left( \lambda ^{2}\right) 
}{i\lambda +m_{+}\left( \lambda ^{2}\right) }=-1+\frac{2i\lambda }{i\lambda
+m_{+}\left( \lambda ^{2}\right) }  \label{L+}
\end{equation}%
in place of $\left( \ref{R}\right) $. We believe that $\left( \ref{R}\right) 
$ is originally due to Faddeev but we could not locate the paper where this
appeared first. When $V$ is supported on the whole line, a similar approach
was used in \cite{GNP97} and \cite{GS97} to define certain relative
reflection coefficient in situations when there is no classical scattering.


\begin{example}
If $V\left( x\right) =-h^{2}\chi_{-}(x)$ then $m_{-}\left( \lambda
^{2}\right) =i\sqrt{\lambda ^{2}+h^{2}}$ and hence $\left( \ref{R}\right) $
takes the form 
\begin{equation*}
R_{-}\left( \lambda \right) =\frac{\lambda -\sqrt{\lambda ^{2}+h^{2}}}{%
\lambda +\sqrt{\lambda ^{2}+h^{2}}}=-\left( \frac{h}{\lambda +\sqrt{\lambda
^{2}+h^{2}}}\right) ^{2}
\end{equation*}%
which is analytic on $\mathbb{C}_{+}\mathbb{\setminus }\left[ 0,ih\right] $.
\end{example}


\subsection{Potentials supported on the full line}

We now define the right reflection coefficient $R$ for any potential subject
to Hypothesis \ref{H1}.


\begin{definition}
\label{Def_refl}Let $V=V_{-}+V_{+}$ where $V_{-}$ is arbitrary (i.e. subject
to Hypothesis \ref{H1}) and $V_{+}\in L_{+}^{1}$. Let $V_{b}:=V\chi _{\left(
-b,\infty \right) }$ and let $R_{b}$ be the right reflection coefficient. We
call the limit 
\begin{equation}
R=\text{w}-\lim R_{b},b\rightarrow \infty ,  \label{defR}
\end{equation}%
if it exists, the right reflection coefficient from $V$.
\end{definition}


\begin{lemma}
\label{Fragmentation}The right reflection coefficient $R\left( \lambda
\right) $ defined by $\left( \ref{defR}\right) $ is well-defined and
satisfies%
\begin{align}
R& =\frac{T_{+}}{\overline{T}_{+}}\dfrac{R_{-}-\overline{L}_{+}}{1-R_{-}L_{+}%
},  \label{fullRC'} \\
\text{or }R& =R_{+}+\dfrac{T_{+}^{2}R_{-}}{1-R_{-}L_{+}}  \label{fullRC}
\end{align}%
where the subscript $\pm $ indicates that the corresponding scattering
quantities are related to $V_{\pm }$, the right hand side of $\left( \ref%
{fullRC}\right) $ being independent of a particular partition $\left( \ref%
{Vfragm}\right) $. Moreover 
\begin{equation*}
\left\vert R\left( \lambda \right) \right\vert \leq 1\text{ for a.e. real }%
\lambda ,
\end{equation*}%
\begin{equation*}
\left\vert R\left( \lambda \right) \right\vert <1\text{ for a.e. real }%
\lambda \in \func{Spec}_{ac}\left( H\right) \text{ of multiplicity }2,
\end{equation*}%
\begin{equation*}
\text{and }\left\vert R\right\vert =1\text{ if and only if }\left\vert
R_{-}\right\vert =1.
\end{equation*}
\end{lemma}

\begin{proof}
To avoid the subscript $b$ we denote $\widetilde{V}:=V_{b}=V\chi _{\left(
-b,\infty \right) }$. The potential splitting 
\begin{equation*}
\widetilde{V}=\widetilde{V}_{-}+V_{+}
\end{equation*}%
implies the fragmentation principle (see, e.g. \cite{AK01}) 
\begin{equation}
\left( 
\begin{array}{cc}
1/\widetilde{T} & -\widetilde{R}/\widetilde{T} \\ 
\widetilde{L}/\widetilde{T} & \overline{1/\widetilde{T}}%
\end{array}%
\right) =\left( 
\begin{array}{cc}
1/\widetilde{T}_{-} & -\widetilde{R}_{-}/\widetilde{T}_{-} \\ 
\widetilde{L}_{-}/\widetilde{T}_{-} & \overline{1/\widetilde{T}}_{-}%
\end{array}%
\right) \left( 
\begin{array}{cc}
1/T_{+} & -R_{+}/T_{+} \\ 
L_{+}/T_{+} & 1/\overline{T}_{+}%
\end{array}%
\right)  \label{fragm}
\end{equation}%
where each entry is well-defined. Multiplying out the matrices in $\left( %
\ref{fragm}\right) $%
\begin{align*}
\frac{1}{\widetilde{T}}& =\frac{1-L_{+}\widetilde{R}_{-}}{T_{+}\widetilde{T}%
_{-}} \\
\frac{\widetilde{R}}{\widetilde{T}}& =\frac{R_{+}}{T_{+}\widetilde{T}_{-}}+%
\frac{\widetilde{R}_{-}}{\overline{T}_{+}\widetilde{T}_{-}},
\end{align*}%
a straightforward algebra yields%
\begin{equation}
\widetilde{R}=\frac{R_{+}+\left( T_{+}/\overline{T}_{+}\right) \widetilde{R}%
_{-}}{1-L_{+}\widetilde{R}_{-}}.  \label{Rtilda}
\end{equation}%
Inserting the following relation from $\left( \ref{2'}\right) $ 
\begin{equation*}
\overline{T}_{+}/T_{+}=-\overline{L}_{+}/R_{+}
\end{equation*}%
into $\left( \ref{Rtilda}\right) $ yields $\left( \ref{fullRC'}\right) $.
Using the relations $\left( \ref{2-3}\right) ,$ and $\left( \ref{2'}\right) $
we have%
\begin{eqnarray*}
\frac{L_{+}R_{+}-T_{+}^{2}}{T_{+}} &=&-\frac{\left( T_{+}/\overline{T}%
_{+}\right) R_{+}\overline{R}_{+}+T_{+}^{2}}{T_{+}} \\
&=&-\frac{\left( T_{+}/\overline{T}_{+}\right) \left( 1-T_{+}\overline{T}%
_{+}\right) +T_{+}^{2}}{T_{+}}=-\frac{1}{\overline{T}_{+}}
\end{eqnarray*}%
which implies that 
\begin{equation*}
T_{+}/\overline{T}_{+}=T_{+}^{2}-L_{+}R_{+}.
\end{equation*}%
Inserting this relation into $\left( \ref{Rtilda}\right) $ one obtains 
\begin{equation}
\widetilde{R}=R_{+}+\widetilde{G},\ \ \ \ \ \ \widetilde{G}:=\dfrac{T_{+}^{2}%
\widetilde{R}_{-}}{1-\widetilde{R}_{-}L_{+}}.  \label{Rtilda1}
\end{equation}%
As discussed in the previous subsections each quantity in $\widetilde{G}$
can admits an analytic continuation into $\mathbb{C}_{+}$ and hence $%
\widetilde{G}$ \ can also be continued into $\mathbb{C}_{+}$. But according
to \cite{PavSmirn82}, $\widetilde{m}_{-}$ converges uniformly on every
compact in $\mathbb{C}_{+}$ to $m_{-}$ as $b\rightarrow \infty $ and hence
so does $\widetilde{R}_{-}$ to $R$. This means that uniformly on every
compact in $\mathbb{C}_{+}$ 
\begin{equation}
\lim_{b\rightarrow \infty }\widetilde{G}=\dfrac{T_{+}^{2}R_{-}}{1-R_{-}L_{+}}%
=:G.  \label{limG}
\end{equation}%
$\left( \ref{limG}\right) $ implies that on the real line $\widetilde{R}%
_{-}\rightarrow R$ and $\widetilde{G}\rightarrow G$ weakly as $b\rightarrow
\infty $. Since $\widetilde{R}$ in $\left( \ref{Rtilda1}\right) $ is
independent of the point of splitting, the reflection coefficient defined by 
$\left( \ref{defR}\right) $ is well-defined and $\left( \ref{fullRC}\right) $
holds. The last statements of the lemma immediately follow from 
\begin{equation*}
1-\left\vert R\right\vert ^{2}=\frac{\left( 1-\left\vert R_{-}\right\vert
^{2}\right) \left( 1-\left\vert L_{+}\right\vert ^{2}\right) }{\left\vert
1-R_{-}L_{+}\right\vert ^{2}}
\end{equation*}%
which in turn follows from the fact that $\left( \ref{fullRC'}\right) $
represents a M\"{o}bius transform (or can be verified by a direct
computation).
\end{proof}

It is quite clear that our definition \ref{Def_refl} agrees with the
standard one if $V$ is short range.

Since each of the scattering quantities in the second term on the right hand
side of $\left( \ref{fullRC}\right) $ can be analytically extended into the
upper half plane, so can the whole second term in $\left( \ref{fullRC}%
\right) $. We have no grounds to believe that $R_{+}$ is analytic under the
Faddeev condition only. However it is the case if $V$ is reflectionless.
Namely, the following curious statement holds.

\begin{proposition}
\label{Reflectionless}Let $V$ be a reflectionless potential (i.e. $R\left(
\lambda \right) =L\left( \lambda \right) =0$) and $V\in L^{1}$then all (left
and right) reflection coefficients corresponding to $V_{\pm }$ can be
analytically continued into the upper half plane.
\end{proposition}

Indeed, it immediately follows from $\left( \ref{fullRC}\right) $ that 
\begin{equation*}
R_{+}=-\dfrac{T_{+}^{2}R_{-}}{1-R_{-}L_{+}}
\end{equation*}%
and hence $R_{+}$ admits an analytic continuation into $\mathbb{C}_{+}$.
Similarly one proves that $L_{-}$ has the same property.

Proposition \ref{Reflectionless} is, of course, well known for $N$-soliton
reflectionless potentials but appears to be new as it admits infinitely many
solitons.


\section{Fredholm determinants}

In this section we specify in what sense we understand Fredholm determinants
(Definition \ref{Def Fred}) and prove two important lemmas.

It is well-known that if $A$ is a trace class operator than one can define
the invariant\footnote{%
i.e. independent of a matrix representation of $A$.} Fredholm determinant $%
\det \left( I+A\right) $. If $A$ is Hilbert-Schmidt then 
\begin{equation*}
\det\nolimits_{2}\left( I+A\right) :=\det \left( I+A\right) e^{-A}
\end{equation*}%
is also well-defined. Apparently if $A\in \mathfrak{S}_{1}$ then 
\begin{equation*}
\det \left( I+A\right) =\det\nolimits_{2}\left( I+A\right) \cdot e^{\limfunc{%
tr}A}.
\end{equation*}

In particular situations it is usually very hard to verify that $A\in 
\mathfrak{S}_{1}$. It is not easy even if $A$ is an integral whereas
verifying $A\in \mathfrak{S}_{2}$ merely requires computing a double
integral. However for an integral operator on $L^{2}\left( S\right) $, $%
S\subseteq \mathbb{R}$, with the kernel $A\left( x,y\right) $ the condition $%
A\left( x,x\right) \in L^{1}\left( S\right) $ is much easier to check and
the trace can then be conveniently defined as the integral of the kernel on
the diagonal. Namely, one introduces


\begin{definition}
\label{Def Fred}Let $A:L^{2}\left( S\right) \rightarrow L^{2}\left( S\right) 
$ be a Hilbert-Schmidt integral operator with the kernel $A\left( x,y\right) 
$. We call $A$ a trace type operator if $A\left( x,x\right) $ is
well-defined, continuous on $S$ and $A\left( x,x\right) \in L^{1}\left(
S\right) $ and we set then 
\begin{eqnarray*}
\limfunc{Tr}A &:&=\int_{S}A\left( x,x\right) dx \\
\limfunc{Det}\left( I+A\right) &:&=\det\nolimits_{2}\left( I+A\right) \cdot
e^{\limfunc{Tr}A}.
\end{eqnarray*}
\end{definition}

Of course if $A\in \mathfrak{S}_{1}$ then 
\begin{equation*}
\limfunc{Det}\left( I+A\right) =\det \left( I+A\right)
\end{equation*}%
but there are examples of Hilbert-Schmidt integral operators $A$'s not from $%
\mathfrak{S}_{1}$ but for which $A\left( x,x\right) \in L^{1}\left( S\right) 
$. Such examples are quite pathological though (see \cite{SimonTrace}).

The Fredholm determinants in two particular cases will be important in our
consideration.


\begin{lemma}
\label{Trace}Let ${\phi }\left( \lambda \right) $ be a smooth function
defined on a piecewise differentiable contour $\Gamma =\left\{ \lambda \in 
\mathbb{C}:\lambda =\alpha +ih\left( \alpha \right) ,h\geq 0,\alpha \in 
\mathbb{R}\right\} $ such that $\dfrac{{\phi }\left( \lambda \right) }{\func{%
Im}\lambda }\in $ $L^{1}\left( \Gamma \right) $. Then the integral operator $%
\mathbf{\Phi}$ on $L_{+}^{2}$ with the kernel 
\begin{equation*}
{\Phi }\left( x,y\right) =\int_{\Gamma }e^{i\lambda \left( x+y\right) }{\phi 
}\left( \lambda \right) \frac{d\lambda }{2\pi },\: x,y\geq 0,
\end{equation*}%
is trace class,%
\begin{equation*}
\left\Vert {{\mathbf{\Phi }}}\right\Vert _{\mathfrak{S}_{1}}\leq \frac{1}{%
4\pi }\left\Vert \frac{{\phi }\left( \lambda \right) }{\func{Im}\lambda }%
\right\Vert _{L^{1}\left( \Gamma \right) }\mathbf{,}
\end{equation*}%
and hence $\det \left( I+\mathbf{\Phi }\right) $ is well-defined in the
classical Fredholm sense.
\end{lemma}

\begin{proof}
Set $\gamma \left( \alpha \right) :=\alpha +ih\left( \alpha \right) $ 
. We have ($x,y\geq 0$)%
\begin{equation}
{\Phi }\left( x,y\right) =\int e^{i\gamma \left( \alpha \right) \left(
x+y\right) }{\phi }\left( \gamma \left( \alpha \right) \right) \gamma
^{\prime }\left( \alpha \right) \frac{d\alpha }{2\pi }.  \label{Int1}
\end{equation}%
We now split this integral in a certain way. To this end let ${\phi }={\phi }%
_{1}{\phi }_{2}$ be a factorization of ${\phi }$ to be chosen later 
\begin{align}
& e^{i\gamma \left( \alpha \right) \left( x+y\right) }{\phi }\left( \gamma
\left( \alpha \right) \right) \gamma ^{\prime }\left( \alpha \right) ^{1/2}
\label{Int2} \\
& =e^{i\gamma \left( \alpha \right) x}{\phi }_{1}\left( \gamma \left( \alpha
\right) \right) \int e^{i\gamma \left( \beta \right) y}{\phi }_{2}\left(
\gamma \left( \beta \right) \right) \delta \left( \beta -\alpha \right)
\gamma ^{\prime }\left( \beta \right) ^{1/2}d\beta  \notag \\
& =\int \left( e^{i\gamma \left( \alpha \right) x}{\phi }_{1}\left( \gamma
\left( \alpha \right) \right) \int e^{i\gamma \left( \beta \right) y}{\phi }%
_{2}\left( \gamma \left( \beta \right) \right) e^{i\left( \beta -\alpha
\right) s}\gamma ^{\prime }\left( \beta \right) ^{1/2}d\beta \right) \dfrac{%
ds}{2\pi }  \notag \\
& ={\phi }_{1}\left( \gamma \left( \alpha \right) \right) \int e^{i\left(
\gamma \left( \alpha \right) x-\alpha s\right) }\left( \int e^{i\left(
\gamma \left( \beta \right) y+\beta s\right) }{\phi }_{2}\left( \gamma
\left( \beta \right) \right) \gamma ^{\prime }\left( \beta \right)
^{1/2}d\beta \right) \dfrac{ds}{2\pi }  \notag
\end{align}%
where we have used%
\begin{equation*}
\delta \left( \beta -\alpha \right) =\int e^{i\left( \beta -\alpha \right) s}%
\dfrac{ds}{2\pi }.
\end{equation*}%
Substituting $\left( \ref{Int2}\right) $ into $\left( \ref{Int1}\right) $
yields%
\begin{align}
{\Phi }\left( x,y\right) & =\int \left\{ \left( \int e^{i\left( \gamma
\left( \alpha \right) x-\alpha s\right) }{\phi }_{1}\left( \gamma \left(
\alpha \right) \right) \gamma ^{\prime }\left( \alpha \right) ^{1/2}\frac{%
d\alpha }{2\pi }\right) \right.  \notag \\
& \left. \left( \int e^{i\left( \gamma \left( \beta \right) y+\beta s\right)
}{\phi }_{2}\left( \gamma \left( \beta \right) \right) \gamma ^{\prime
}\left( \beta \right) ^{1/2}\frac{d\beta }{2\pi }\right) \right\} ds  \notag
\\
& =\int {\Phi }_{1}\left( x,s\right) {\Phi }_{2}\left( s,y\right) ds
\label{Int3}
\end{align}%
where 
\begin{align*}
{\Phi }_{1}\left( x,s\right) & =\int \exp \left\{ i\left( x-s\right) \alpha
-h\left( \alpha \right) x\right\} {\phi }_{1}\left( \gamma \left( \alpha
\right) \right) \gamma ^{\prime }\left( \alpha \right) ^{1/2}d\alpha \\
{\Phi }_{2}\left( s,y\right) & =\int \exp \left\{ i\left( s+y\right) \alpha
-h\left( \alpha \right) y\right\} {\phi }_{2}\left( \gamma \left( \alpha
\right) \right) \gamma ^{\prime }\left( \alpha \right) ^{1/2}d\alpha .
\end{align*}%
By a straightforward computation ($k=1,2$) 
\begin{align*}
\iint_{\mathbb{R}_{+}} \left\vert {\Phi }_{k}\left( x,y\right) \right\vert
^{2}dxdy& =\int \frac{\left\vert \gamma ^{\prime }\left( \alpha \right)
\right\vert }{2h\left( \alpha \right) }\left\vert {\phi }_{k}\left( \gamma
\left( \alpha \right) \right) \right\vert ^{2}\frac{d\alpha }{2\pi } \\
& =\int_{\Gamma }\left\vert \frac{{\phi }_{k}\left( \lambda \right) }{\sqrt{%
\func{Im}\lambda }}\right\vert ^{2}\frac{\left\vert d\lambda \right\vert }{%
4\pi }.
\end{align*}%
On the other hand%
\begin{equation*}
\iint_{\mathbb{R}_{+}} \left\vert {\Phi }_{k}\left( x,y\right) \right\vert
^{2}dxdy=\left\Vert \mathbf{\Phi }_k\right\Vert _{\mathfrak{S}_{2}}^{2}
\end{equation*}%
and hence 
\begin{equation*}
\left\Vert {{\mathbf{\Phi }}}_{k}\right\Vert _{\mathfrak{S}_{2}}=\frac{1}{2%
\sqrt{\pi }}\left\Vert \frac{{\phi }_{k}\left( \lambda \right) }{\sqrt{\func{%
Im}\lambda }}\right\Vert _{L^{2}\left( \Gamma \right) }.
\end{equation*}%
It follows from $\left( \ref{Int3}\right) $ that ${{\mathbf{\Phi }}}={{%
\mathbf{\Phi }}}_{1}{{\mathbf{\Phi }}}_{2}$ and, taking ${\phi }_{1}={\phi }%
_{2}$ we finally have%
\begin{eqnarray*}
\left\Vert {{\mathbf{\Phi }}}\right\Vert _{\mathfrak{S}_{1}} &\leq
&\left\Vert {{\mathbf{\Phi }}}_{1}\right\Vert _{\mathfrak{S}_{2}}\left\Vert {%
{\mathbf{\Phi }}}_{2}\right\Vert _{\mathfrak{S}_{2}}=\frac{1}{4\pi }%
\left\Vert \sqrt{\frac{{\phi }\left( \lambda \right) }{\func{Im}\lambda }}%
\right\Vert _{L^{2}\left( \Gamma \right) }^{2} \\
&=&\frac{1}{4\pi }\left\Vert \frac{{\phi }\left( \lambda \right) }{\func{Im}%
\lambda }\right\Vert _{L^{1}\left( \Gamma \right) }<\infty
\end{eqnarray*}%
and the lemma is proven.
\end{proof}


\begin{lemma}
\label{Trace type}Let $R$ be as in (\ref{R(lambda/2)}) such that $\partial
_{\lambda }R$ $\in L^{2}$ and $t\geq 0$ is a parameter (time), Then the
integral operator $\mathbf{\Phi }$ on $L^{2}\left( a,\infty \right)
,a>-\infty ,$ with the kernel%
\begin{equation*}
{\Phi }\left( x,y\right) =\left( \mathcal{F}e^{8i\lambda ^{3}t}R\right)
\left( x+y\right)
\end{equation*}%
is a trace type operator in the sense of Definition\ \ref{Def Fred}.
\end{lemma}

\begin{proof}
The fact that $\mathbf{\Phi \in }\mathfrak{S}_{2}$ is well known \cite%
{Deift79} and, since $R$ is in $L^{1}$, ${\Phi }\left( x,x\right) $ is
clearly continuous on $\mathbb{R}$. We only need to show that ${\Phi }\left(
x,x\right) $ decays fast enough to guarantee ${\Phi }\left( x,x\right) \in
L^{1}\left( b,\infty \right) $ for large $b$'s. For $t=0$ the statement is
obvious and it is enough to assume $t>0$. By parts%
\begin{eqnarray*}
2\sqrt{2\pi }\left( \mathcal{F}e^{8i\lambda ^{3}t}R\right) \left( 2x\right)
&=&\int R\left( \lambda /2\right) \dfrac{de^{i(\lambda x+\lambda ^{3}t)}}{%
i(x+3\lambda ^{2}t)} \\
&=&i\int e^{i(\lambda x+\lambda ^{3}t)}\partial _{\lambda }\dfrac{R\left(
\lambda /2\right) }{(x+3\lambda ^{2}t)}d\lambda \\
&=&-i\int e^{i(\lambda x+\lambda ^{3}t)}\dfrac{6\lambda tR\left( \lambda
/2\right) }{(x+3\lambda ^{2}t)^{2}}d\lambda \\
&& +i\int e^{i(\lambda x+\lambda ^{3}t)}\dfrac{\partial _{\lambda }R\left(
\lambda /2\right) }{x+3\lambda ^{2}t}d\lambda \\
&=:&I_{1}\left( 2x\right) +I_{2}\left( 2x\right) .
\end{eqnarray*}%
We have%
\begin{eqnarray*}
\left\Vert I_{1}\right\Vert _{L^{1}\left( b,\infty \right) } &\leq &6t\int
\left\{ \int_{b}^{\infty }\dfrac{dx}{(x+3\lambda ^{2}t)^{2}}\right\}
\left\vert \lambda R\left( \lambda /2\right) \right\vert d\lambda \\
&=&6t\int \dfrac{\left\vert \lambda R\left( \lambda /2\right) \right\vert }{%
b+3\lambda ^{2}t}d\lambda \leq 12t\left\Vert R\right\Vert _{L^{2}}\left\Vert 
\dfrac{\lambda }{b+3\lambda ^{2}t}\right\Vert _{L^{2}\left( d\lambda \right)
}.
\end{eqnarray*}%
Turn now to $I_{2}$. Observe that the following convolution type formula
holds\footnote{%
With the usual definition of the convolution%
\begin{equation*}
\left( f\ast g\right) \left( x\right) =\frac{1}{\sqrt{2\pi }}\int
f\left(s\right) g\left(x- s\right) ds.
\end{equation*}%
}:%
\begin{equation}
\left( \mathcal{F}f_{x}g\right) \left( x\right) =\left( \mathcal{F}f_{x}\ast 
\mathcal{F}g\right) \left( x\right)  \label{convol}
\end{equation}%
where the subscript $x$ indicates that $f$ depends on $x$. Rewriting 
\begin{equation*}
I_{2}\left( 2x\right) =i\int e^{i\lambda x}\left\{ \dfrac{1}{x+3\lambda ^{2}t%
}\right\} \left\{ e^{i\lambda ^{3}t}\partial _{\lambda }R\left( \lambda
/2\right) \right\} d\lambda
\end{equation*}%
and applying the convolution formula $\left( \ref{convol}\right) $ we have 
\begin{equation}
I_{2}\left( 2x\right) =i\sqrt{2\pi }F\left( \cdot ,x\right) \ast \mathcal{F}%
\left( e^{i\lambda ^{3}t}\partial _{\lambda }R\right)  \label{convol 1}
\end{equation}%
where 
\begin{eqnarray*}
F\left( s,x\right) &=&\frac{1}{\sqrt{2\pi }}\int \frac{e^{i\lambda s}}{%
x+3\lambda ^{2}t}d\lambda \\
&=&\sqrt{2\pi }\frac{e^{-\left\vert s\right\vert \left( x/3t\right) ^{1/2}}}{%
6t\left( x/3t\right) ^{1/2}}.
\end{eqnarray*}%
It follows from $\left( \ref{convol 1}\right) $ that 
\begin{eqnarray}
\left\Vert I_{2}\right\Vert _{L^{1}\left( b,\infty \right) } &\leq &\int
\left\Vert F\left( s,x\right) \mathcal{F}\left( e^{i\lambda ^{3}t}\partial
_{\lambda }R\right) \left( x-s\right) \right\Vert _{L^{1}\left( b,\infty
\right) }ds  \notag \\
&\leq &\int \left\Vert F\left( s,\cdot \right) \right\Vert _{L^{2}\left(
b,\infty \right) }\left\Vert \mathcal{F}\left( e^{i\lambda ^{3}t}\partial
_{\lambda }R\right) \right\Vert _{L^{2}\left( b-s,\infty \right) }ds  \notag
\\
&\leq &\left\Vert \mathcal{F}\left( e^{i\lambda ^{3}t}\partial _{\lambda
}R\right) \right\Vert _{L^{2}}\int \left\Vert F\left( s,\cdot \right)
\right\Vert _{L^{2}\left( b,\infty \right) }ds  \notag \\
&=&\left\Vert \partial _{\lambda }R\right\Vert _{L^{2}}\int \left\Vert
F\left( s,\cdot \right) \right\Vert _{L^{2}\left( b,\infty \right) }ds.
\label{estim}
\end{eqnarray}%
The norm $\left\Vert F\left( s,\cdot \right) \right\Vert _{L^{2}\left(
b,\infty \right) }$ on the right hand side of $\left( \ref{estim}\right) $
can be explicitly evaluated:%
\begin{eqnarray*}
\left\Vert F\left( s,\cdot \right) \right\Vert _{L^{2}\left( b,\infty
\right) }&=&\frac{\sqrt{2\pi }}{6t}\left( \int_{b}^{\infty }\frac{%
e^{-2\left\vert s\right\vert \left( x/3t\right) ^{1/2}}}{x/3t}dx\right)
^{1/2} \\
&=&\sqrt{\frac{\pi }{6t}}\left( \int_{b}^{\infty }\frac{e^{-2\left\vert
s\right\vert \left( x/3t\right) ^{1/2}}}{x}dx\right) ^{1/2} \\
&=&\sqrt{\frac{\pi }{3t}}\left( \int_{2\sqrt{\frac{b}{3t}}\left\vert
s\right\vert }^{\infty }\frac{e^{-x}}{x}dx\right) ^{1/2},
\end{eqnarray*}%
which immediately implies that $\left\Vert F\left( s,\cdot \right)
\right\Vert _{L^{2}\left( b,\infty \right) }$ is continuous with respect to $%
s\in \mathbb{R}\backslash \left\{ 0\right\} $ and 
\begin{equation}
\left\Vert F\left( s,\cdot \right) \right\Vert _{L^{2}\left( b,\infty
\right) }=o\left( e^{-\sqrt{\frac{b}{3t}}\left\vert s\right\vert }\right)
,\: s\rightarrow \pm \infty .  \label{large s}
\end{equation}%
Around $s=0$ (denoting $\alpha :=2\sqrt{\frac{b}{3t}}\left\vert s\right\vert 
$) 
\begin{equation*}
\int_{\alpha }^{\infty }\frac{e^{-x}}{x}dx=\frac{e^{-\alpha }}{\alpha }%
-\int_{\alpha }^{\infty }\frac{e^{-x}}{x^{2}}dx=O\left( 1/\alpha \right)
=O\left( 1/\left\vert s\right\vert \right) ,\: s\rightarrow 0,
\end{equation*}%
and therefore%
\begin{equation}
\left\Vert F\left( s,\cdot \right) \right\Vert _{L^{2}\left( b,\infty
\right) }=O\left( 1/\left\vert s\right\vert ^{1/2}\right) ,\: s\rightarrow 0.
\label{small s}
\end{equation}%
$\left( \ref{large s}\right) $ and $\left( \ref{small s}\right) $ imply that 
$\int \left\Vert F\left( s,\cdot \right) \right\Vert _{L^{2}\left( b,\infty
\right) }ds$ is finite and the lemma is proven since $\left\Vert \partial
_{\lambda }R\right\Vert _{L^{2}}$ is also finite.
\end{proof}


\section{A Hankel Integral Operator}

In this section we introduce and study a Hankel integral operator
particularly important in the context of the IST.


\begin{definition}
\label{March type}Let $\mu $ be a non-negative finite measure on $\mathbb{R}%
_{+}$ and $\phi $ be an $L^{\infty }$ function. We call an operator $\mathbb{%
M}:L_{+}^{2}\rightarrow L_{+}^{2}$ a Marchenko type operator associated with 
$\left( \mu ,\phi \right) $ if%
\begin{equation}
\mathbb{M}=\mathbb{M}_{1}+\mathbb{M}_{2}  \label{K}
\end{equation}%
where $\mathbb{M}_{1}$ is the integral operator with the kernel 
\begin{equation}
M_{1}\left( x,y\right) =\int_{\mathbb{R}_{+}}e^{-\alpha \left( x+y\right)
}d\mu \left( \alpha \right)  \label{K1}
\end{equation}%
and ($\chi :=\chi _{+}$) 
\begin{equation}
\mathbb{M}_{2}=\chi \mathcal{F}\phi \mathcal{F}.  \label{K2}
\end{equation}%
Here $\phi $ and $\chi $ are the operators of multiplication by the
functions $\phi $ and $\chi $ respectively.
\end{definition}

The Marchenko operator $\mathbb{M}_{x,t}$ defined by $\left( \ref{ClassM'}%
\right)-\eqref{classM} $ is Marchenko type as it can be represented by $%
\left( \ref{K}\right) -\left( \ref{K2}\right) $ with%
\begin{eqnarray*}
d\mu \left( \alpha \right) &=&\dsum_{n=1}^{N}c_{n}^{2}e^{-2\alpha x+8\alpha
^{3}t}\delta \left( \alpha -\kappa _{n}\right) d\alpha , \\
\phi \left( \lambda \right) &=&e^{2i\lambda x+8i\lambda ^{3}t}R(\lambda ),
\end{eqnarray*}%
where $\delta $ denotes the Dirac delta function.

The operator $\mathbb{M}$ is clearly a Hankel operator. In this section we
are concerned with two main questions: when is $\mathbb{M}$ a trace class
operator (or at least when is $\limfunc{Det}\left( I+\mathbb{M}\right) $
well-defined) and when is $I+\mathbb{M}$ boundedly invertible?

Introduce yet another two parametric family ( $z\in \mathbb{C}$ and $t\geq 0$
are parameters) of integral operators 
\begin{equation}
\left( \mathbb{G}_{z,t}f\right) \left( x\right) :=\int_{\mathbb{R_{+}}}
G_{z,t}\left( x,y\right) f\left( y\right) dy,  \label{G-oper}
\end{equation}%
acting in $L_{+}^{2}$ with the kernel $G_{z,t}\left( x,y\right) $ defined by 
\begin{equation}
G_{z,t}\left( x,y\right) :=\int_{\Gamma }e^{i\lambda \left( x+y\right)
}g_{z,t}\left( \lambda \right) \frac{d\lambda }{2\pi }  \label{Mkern}
\end{equation}%
where $\Gamma $ is as in the Figure \ref{fig1} and

\begin{figure}[ht]
\begin{tikzpicture}[scale=0.6]
\draw[help lines] (-7.5,0) -- (7.5,0)
			(0,-0.3) -- (0,4.5);
\draw (0,0) node[below left]{$0$}
			(-0.05,1.5) -- + (0.1,0) node[right]{$ih_0$};
			
\draw[thick] (150:8) -- (150:4.2) -- (30:4.2) -- node[above]{$\Gamma$}(30:8); 
\draw[dashed] (150:4) -- (0,0) -- (30:4);
\draw (1,0) arc (0:30:1);
\draw (1.3,0.5) node[right]{$\pi/6$};

\end{tikzpicture}
\caption{Contour of integration $\Gamma $ in equation \eqref{Mkern}, where $%
h_{0}=\left\vert \inf \limfunc{Spec} (H)\right\vert ^{1/2}.$}
\label{fig1}
\end{figure}
\begin{equation}
g_{z,t}(\lambda ):=e^{2i\lambda z}e^{8i\lambda ^{3}t}G(\lambda )  \label{g}
\end{equation}%
with some function $G$ specified in the proposition below. The operator $%
\mathbb{G}_{z,t}$ is a Hankel operator having some important properties
which we summarize in the following statement.


\begin{proposition}
\label{Properties of G} Let $\mathbb{G}_{z,t}:L_{+}^{2}\rightarrow L_{+}^{2}$
be defined by $\left( \ref{G-oper}\right) $-$\left( \ref{Mkern}\right) $
with some $G$ analytic in $\mathbb{C}_{+}\diagdown \left[ 0,ia\right] ,a\geq
0$, subject to

\begin{enumerate}
\item \renewcommand{\theenumi}{\roman{enumi}}

\item (symmetry) 
\begin{equation}
G\left( -\overline{\lambda }\right) =\overline{G\left( \lambda \right) }
\label{symmetry}
\end{equation}

\item (decay)%
\begin{equation*}
\left\vert G\left( \lambda \right) \right\vert \rightarrow 0,\: \left\vert
\lambda \right\vert \rightarrow \infty ,\: 0<\arg \lambda <\pi
\end{equation*}

\item (boundary values on the real line)%
\begin{equation*}
G\left( \lambda +i0\right) \in L^{\infty }
\end{equation*}

\item (boundary behavior on the imaginary line) 
\begin{equation}
d\rho \left( \alpha \right) :=\frac{1}{\pi }\func{Im}G\left( +0+i\alpha
\right) d\alpha  \label{drou}
\end{equation}%
defines a non-negative finite on $\left[ 0,a\right] $ measure. I.e.%
\begin{equation*}
d\rho \left( \alpha \right) \geq 0\text{ and }\int_{0}^{a}d\rho \left(
\alpha \right) <\infty .
\end{equation*}
\end{enumerate}

Then

\begin{enumerate}
\item \label{pr7.2.1} $\mathbb{G}_{z,t}$ is a Marchenko type operator
(Definition \ref{March type}) associated with $\left( \mu ,\phi \right) $
given by%
\begin{equation*}
d\mu \left( \alpha \right) =e^{-2\alpha z+8\alpha ^{3}t}d\rho \left( \alpha
\right) \text{ and }\phi \mathbb{=}g_{z,t}.
\end{equation*}%
Moreover

\item \label{pr7.2.2} $\mathbb{G}_{z,t}$ is selfadjoint for any $z\in 
\mathbb{R}$ and $t\geq 0$

\item \label{pr7.2.3} $\mathbb{G}_{z,t}\in \mathfrak{S}_{1}$ for any $z\in 
\mathbb{C}$ and $t>0$ and%
\begin{equation*}
\left\Vert \mathbb{G}_{z,t}\right\Vert _{\mathfrak{S}_{1}}\leq \frac{1}{4\pi 
}\left\Vert \frac{g_{z,t}(\lambda )}{\func{Im}\lambda }\right\Vert _{\mathrm{%
L}^{1}(\Gamma )}
\end{equation*}

\item \label{pr7.2.4} $\mathbb{G}_{z,t}$ is entire and $\left( I+\mathbb{G}%
_{z,t}\right) ^{-1}$ is a meromorphic operator valued function in $z$ on the
entire complex plane for any $t>0$.
\end{enumerate}
\end{proposition}

\begin{proof}
Due to $\left( \ref{symmetry}\right) $ 
the kernel $\left( \ref{Mkern}\right) $ is real 
for real $z$ and symmetric. The operator $\mathbb{G}_{z,t}$ is therefore
selfadjoint for real $z$ and Part (\ref{pr7.2.2}) is proven. To prove the
representation $\left( \ref{K}\right) -\left( \ref{K2}\right) $ one merely
needs to deform the contour $\mathbb{\Gamma }$ to the real line. The only
issue is to make sure that the corresponding integral operators converge
strongly. Denoting ($\Gamma ^{+}:=\mathbb{\Gamma \cap }\left\{ \limfunc{Re}%
z>0\right\} $)%
\begin{equation*}
G_{z,t}\left( x\right) :=\func{Re}\int_{\Gamma ^{+}}e^{i\lambda
x}g_{z,t}\left( \lambda \right) \frac{d\lambda }{\pi },
\end{equation*}%
we have 
\begin{eqnarray}
G_{z,t}\left( x\right) &=&G_{z,t}^{\left( 1\right) }\left( x\right)
+G_{z,t}^{\left( 2\right) }\left( x\right) +G_{z,t}^{\left( 3\right) }\left(
x\right) +G_{z,t}^{\left( 4\right) }\left( x\right)  \label{ReM} \\
G_{z,t}^{\left( k\right) }\left( x\right) &:=&\frac{1}{\pi }\func{Re}\int_{%
\mathbb{\gamma }_{k}}e^{i\lambda x}G_{z,t}\left( \lambda \right) \frac{%
d\lambda }{\pi },\: k=1,2,3,4.  \notag
\end{eqnarray}%
where 
\begin{equation*}
\begin{array}{ll}
\mathbb{\gamma }_{1}=\left\{ i\alpha +0:0<\alpha <a\right\} , & \quad\mathbb{%
\gamma }_{2}=\left( 0,N\right) , \\ 
\mathbb{\gamma }_{3}=\left\{ Ne^{i\theta },0<\theta <\pi /6\right\} , & \quad%
\mathbb{\gamma }_{4}=\left\{ re^{i\pi /6}:r>N\right\}%
\end{array}%
\end{equation*}
as shown in Figure \ref{fig2} 
\begin{figure}[th]
\begin{tikzpicture}
\draw[gray,very thin] (-0.8,0) -- (6,0)
			(0,-0.3) -- (0,2.8);
\draw (0,0) node[below left]{$0$}
			(0,2) node[left]{$ia$}
			(-0.05,2) -- (0.05,2)
			(3.5,0) node[below]{$N$}
			(3.5,-0.05) -- (3.5,0.05);
\begin{scope}[decoration={markings,mark= at position 0.5 with {\arrow{stealth}}}]
\draw[thick,postaction={decorate}] (0,2) -- node[left]{$\gamma_1$} (0,0); 
\draw[thick,postaction={decorate}] (0,0) -- node[below right]{$\gamma_2$} (3.5,0);
\draw[thick,postaction={decorate}] (3.5,0) arc(0:30:3.5); 
\draw (15:3.5) node[right]{$\gamma_3$};
\draw[thick,postaction={decorate}] (30:3.5) -- node[above left]{$\gamma_4$} (30:6);
 \end{scope}
\draw[dashed] (0,0) -- (30:3.5);
\draw (0:1) arc(0:30:1);
\draw (15:1) node[right] {$\pi/6$}; 
\end{tikzpicture}
\caption{{}}
\label{fig2}
\end{figure}

The representation $\left( \ref{ReM}\right) $ leads to%
\begin{equation*}
\mathbb{G}_{z,t}=\mathbb{G}_{z,t}^{\left( 1\right) }+\mathbb{G}%
_{z,t}^{\left( 2\right) }+\mathbb{G}_{z,t}^{\left( 3\right) }+\mathbb{G}%
_{z,t}^{\left( 4\right) }
\end{equation*}%
where $\mathbb{G}_{z,t}^{\left( k\right) }$ are the integral operators
defined by ($f_{+}=f\chi _{+}$) 
\begin{equation*}
\mathbb{G}_{z,t}^{\left( k\right) }f\left( x\right) =\int G_{z,t}^{\left(
k\right) }\left( x+y\right) f_{+}\left( y\right) dy.
\end{equation*}%
The arbitrary $N>0$ will be taken to infinity. For $G_{z,t}^{\left( 1\right)
}\left( x\right) $ we get%
\begin{eqnarray*}
G_{z,t}^{\left( 1\right) }\left( x\right) &=&\frac{1}{\pi }\lim_{\varepsilon
\rightarrow +0}\func{Re}\int_{a}^{0}e^{i\left( \varepsilon +i\alpha \right)
x}g_{z,t}\left( \varepsilon +i\alpha \right) id\alpha \\
&=&\frac{1}{\pi }\lim_{\varepsilon \rightarrow +0}\int_{0}^{a}e^{-\alpha x}%
\func{Im}g_{z,t}\left( \varepsilon +i\alpha \right) d\alpha \\
&=&\frac{1}{\pi }\lim_{\varepsilon \rightarrow +0}\int_{0}^{a}e^{-\alpha
x}e^{-2\alpha z+8\alpha ^{3}t}\func{Im}G\left( \varepsilon +i\alpha \right)
d\alpha \\
&=&\int_{0}^{a}e^{-\alpha x}e^{-2\alpha z+8\alpha ^{3}t}d\rho \left( \alpha
\right) ,
\end{eqnarray*}%
and $G_{z,t}^{\left( 1\right) }$ hence produces $\mathbb{M}_{1}$ in the
decomposition $\left( \ref{K}\right) $. For $\mathbb{G}_{z,t}^{\left(
2\right) }$: ($f_{+}=\chi_{+}f$)%
\begin{eqnarray*}
\mathbb{G}_{z,t}^{\left( 2\right) }f\left( x\right) &=&\int G_{z,t}^{\left(
2\right) }\left( x+y\right) f_{+}\left( y\right) dy \\
&=&\frac{1}{\sqrt{2\pi }}\int_{-N}^{N}e^{i\lambda x}\left( g_{z,t}\mathcal{F}%
f_{+}\right) \left( \lambda \right) d\lambda .
\end{eqnarray*}%
Since $\mathcal{F}f_{+}\in L^{2}$ and $g_{z,t}\in L^{\infty }$ we have $%
g_{z,t}\mathcal{F}f_{+}\in L^{2}$, and hence 
\begin{eqnarray*}
\left\Vert \mathbb{G}_{z,t}f-\mathbb{G}_{z,t}^{\left( 2\right) }f\right\Vert
_{L^{2}} &=&\left\Vert \frac{1}{\sqrt{2\pi }}\int \left( 1-\chi _{\left(
-N,N\right) }\left( \lambda \right) \right) e^{i\lambda x}g_{z,t}\left(
\lambda \right) (\mathcal{F}f_{+})\left( \lambda \right) d\lambda
\right\Vert _{L^{2}} \\
&\lesssim &\left\Vert \int e^{i\lambda x}\left( 1-\chi _{\left( -N,N\right)
}\left( \lambda \right) \right) g_{z,t}\left( \lambda \right) \left( 
\mathcal{F}f_{+}\right) \left( \lambda \right) d\lambda \right\Vert _{L^{2}}
\\
&=&\left\Vert \mathcal{F}\left( 1-\chi _{\left( -N,N\right) }\right)
g_{z,t}\left( \mathcal{F}f_{+}\right) \right\Vert _{L^{2}} \\
&=&\left\Vert \left( 1-\chi _{\left( -N,N\right) }\right) g_{z,t}\mathcal{F}%
\left( f_{+}\right) \right\Vert _{L^{2}}\rightarrow 0,\: N\rightarrow \infty
.
\end{eqnarray*}%
Therefore $\mathbb{G}_{z,t}^{\left( 2\right) }\rightarrow \mathbb{G}_{z,t}$
in the strong operator topology for any real $z$ and $t$. It follows from
Lemma \ref{Trace} that $\left\Vert G_{z,t}^{(4)}\right\Vert _{\mathfrak{S}%
_{2}}\rightarrow 0$ when $N\rightarrow \infty $. It remains to show that $%
\left\Vert \mathbb{G}_{z,t}^{\left( 3\right) }f\right\Vert
_{L^{2}}\rightarrow 0,\: N\rightarrow \infty $. Since $\mathbb{G}_{z,t}$ and 
$\mathbb{G}_{z,t}^{\left( 1\right) }$ are independent of $N$ and, as we have
already proven, $\mathbb{G}_{z,t}^{\left( 2\right) }+\mathbb{G}%
_{z,t}^{\left( 4\right) }$ strongly converges to $\mathbb{G}_{z,t}^{+}$, it
follows from the decomposition $\left( \ref{ReM}\right) $ that $\mathbb{G}%
_{z,t}^{\left( 3\right) }$ must also converge strongly. It is sufficient to
show that $\mathbb{G}_{z,t}^{\left( 3\right) }$ converges weakly to $0$ as $%
N\rightarrow \infty $ as this will force the strong convergence of $\mathbb{G%
}_{z,t}^{\left( 3\right) }$ to $0$. We have ($\xi =e^{i\theta }$)%
\begin{eqnarray}
\left( \mathbb{G}_{z,t}^{\left( 3\right) }f\right) \left( x\right) &=&\func{%
Re}\int_{0}^{\infty }\left( \int_{0}^{\pi /6}e^{iN\xi \left( x+y\right)
}g_{z,t}\left( N\xi \right) iN\xi \frac{d\theta }{\pi }\right) f\left(
y\right) dy  \label{M3} \\
&=&-N\func{Im}\int_{0}^{\pi /6}\left( \int_{0}^{\infty }e^{iN\xi y}f\left(
y\right) dy\right) e^{iN\xi x}g_{z,t}\left( N\xi \right) \xi \frac{d\theta }{%
\pi }  \notag \\
&=&-\sqrt{2}N\func{Im}\int_{0}^{\pi /6}e^{iN\xi x}\left( g_{z,t}\mathcal{F}%
f_{+}\right) \left( N\xi \right) \xi \frac{d\theta }{\sqrt{\pi }}  \notag
\end{eqnarray}%
Changing the order of integration in $\left( \ref{M3}\right) $ is justified
as Fubini's theorem clearly applies. Consider $\left\langle \mathbb{G}%
_{z,t}^{\left( 3\right) }f,\varphi \right\rangle $ with arbitrary $f,\varphi
\in L^{2}$ which, without loss of generality, can be taken real. Hence one
has 
\begin{eqnarray}
\left\vert \left\langle \mathbb{G}_{z,t}^{\left( 3\right) }f,\varphi
\right\rangle \right\vert &=&\sqrt{2}N\left\vert \int_{0}^{\pi /6}\func{Im}%
\left( g_{z,t}\right) \left( \mathcal{F}\left( f_{+}\right) \mathcal{F}%
\left( \varphi _{+}\right) \right) \left( N\xi \right) \xi d\theta
\right\vert  \label{M3'} \\
&\lesssim &N\int_{0}^{\pi /6}\left\vert \left( g_{z,t}\mathcal{F}f_{+}%
\mathcal{F}\varphi _{+}\right) \left( N\xi \right) \right\vert d\theta 
\notag
\end{eqnarray}%
Since $\mathcal{F}f_{+},\mathcal{F}\varphi _{+}$ are both in $H_{+}^{2}$
their product $F:=\left( \mathcal{F}f_{+}\right) \left( \mathcal{F}\varphi
_{+}\right) $ is in $H_{+}^{1}$ and so by the Hardy-Littlewood theorem the
maximal function $F^{\ast }\left( N\right) :=\sup\limits_{\alpha
>0}\left\vert F\left( Ne^{i\alpha} \right) \right\vert $ is in $L^{1}.$
Therefore it follows from $\left( \ref{M3'}\right) $ that 
\begin{eqnarray*}
\left\vert \left\langle \mathbb{G}_{z,t}^{\left( 3\right) }f,\varphi
\right\rangle \right\vert &\lesssim &NF^{\ast }\left( N\right) \int_{0}^{\pi
/6}\left\vert g_{z,t}\left( N\xi \right) \right\vert d\theta \\
&\lesssim &NF^{\ast }\left( N\right) \int_{0}^{\pi /6}\exp \left\{ \frac{4}{%
\pi }N\left( \left\vert z\right\vert -12N^{2}t\right) \theta \right\}
\left\vert g\left( N\xi \right) \right\vert d\theta \\
&\lesssim &\frac{F^{\ast }\left( N\right) }{12N^{2}t-\left\vert z\right\vert 
}.
\end{eqnarray*}

The latter implies that there is a sequence $\left\{ N_{k}\right\}
\rightarrow \infty $ such that $\left\langle \mathbb{G}_{z,t}^{\left(
3\right) }f,\varphi \right\rangle \rightarrow 0$ and Part (\ref{pr7.2.1}) is
finally proven.

Part (\ref{pr7.2.3}), immediately follows from Lemma \ref{Trace}. It only
remains to show Part (\ref{pr7.2.4}). To this end consider for any $f$ and $%
g $ from $L^{2}_{+}$ the function $\left\langle \mathbb{G}%
_{z,t}f,g\right\rangle $ which is clearly differentiable in $z$ for any
complex $z$ and $t>0 $ and hence $\mathbb{G}_{z,t}$ is an entire
operator-valued function. We now show that for any $t>0$ the operator $I+%
\mathbb{G}_{z,t}$ is boundedly invertible for at least one real $z$.
Consider 
\begin{eqnarray*}
\left\langle \mathbb{G}_{z,t}f,f\right\rangle &=&\int_{\mathbb{R}_{+}}%
\overline{f\left( x\right) }dx\int_{\mathbb{R}_{+}}dy\;G_{z,t}\left(
x,y\right) f\left( y\right) \\
&=&\int_{\Gamma }g_{z,t}(\lambda )\left\{ \frac{1}{\sqrt{2\pi }}%
\int_{0}^{\infty }e^{i\lambda x}\overline{f\left( x\right) }dx\right\}
\left\{ \frac{1}{\sqrt{2\pi }}\int_{0}^{\infty }e^{i\lambda y}f\left(
y\right) dy\right\} d\lambda \\
&=&\int_{\Gamma }g_{z,t}(\lambda )\mathcal{F}\overline{f_{+}}\left( \lambda
\right) \mathcal{F}f_{+}\left( \lambda \right) d\lambda =\int_{\Gamma
}g_{z,t}(\lambda )F\left( \lambda \right) d\lambda ,
\end{eqnarray*}%
where $F:=\mathcal{F}\overline{f_{+}}\mathcal{F}f_{+}$. Notice that since $%
\mathcal{F}\overline{f_{+}},\mathcal{F}f_{+}\in H_{+}^{2}$, the function $%
F\in H_{+}^{1}$ and therefore by $\left( \ref{Hp}\right) $ 
\begin{eqnarray*}
\left\vert F\left( \lambda \right) \right\vert &\leq &\frac{\left\Vert
F\right\Vert _{H_{+}^{1}}}{\func{Im}\lambda }\leq \frac{\left\Vert \mathcal{F%
}f_{+}\right\Vert _{H_{+}^{2}}\left\Vert \mathcal{F}\overline{f_{+}}%
\right\Vert _{H_{+}^{2}}}{\func{Im}\lambda } \\
&=&\frac{\left\Vert f\right\Vert_{L^2_{+}} \left\Vert \overline{f}%
\right\Vert_{L^2_{+}} }{\func{Im}\lambda }=\frac{\left\Vert f\right\Vert
^{2}_{L^2_{+}}}{\func{Im}\lambda }.
\end{eqnarray*}%
Hence for $z>0$%
\begin{eqnarray}
\left\vert \left\langle \mathbb{G}_{z,t}f,f\right\rangle \right\vert &\leq
&\int_{\Gamma }\left\vert g_{z,t}(\lambda )\right\vert \left\vert F\left(
\lambda \right) \right\vert \left\vert d\lambda \right\vert  \notag \\
&\leq &\left\Vert f\right\Vert_{L^2_{+}}^{2}\int_{\Gamma }\left\vert
g_{z,t}(\lambda )\right\vert \frac{\left\vert d\lambda \right\vert }{\func{Im%
}\lambda }  \notag \\
&\leq &\left\Vert f\right\Vert_{L^2_{+}}^{2}\sup_{\lambda \in \Gamma
}\left\vert e^{2i\lambda z}\right\vert \int_{\Gamma }\left\vert
g_{0,t}(\lambda )\right\vert \frac{\left\vert d\lambda \right\vert }{\func{Im%
}\lambda }  \notag \\
&=&\left\Vert f\right\Vert_{L^2_{+}}^{2}\sup_{\lambda \in \Gamma }e^{-2z%
\func{Im}\lambda }\int_{\Gamma }\left\vert g_{0,t}(\lambda )\right\vert 
\frac{\left\vert d\lambda \right\vert }{\func{Im}\lambda }  \notag \\
&\leq &\left\Vert f\right\Vert_{L^2_{+}}^{2}e^{-2zh}\int_{\Gamma }\left\vert
g_{0,t}(\lambda )\right\vert \frac{\left\vert d\lambda \right\vert }{\func{Im%
}\lambda },  \label{ineq2}
\end{eqnarray}%
where $h:=\inf_{\lambda \in \Gamma }\func{Im}\lambda $. If we choose (the
integral $\int_{\Gamma }\left\vert g_{0,t}(\lambda )\right\vert \frac{%
\left\vert d\lambda \right\vert }{\func{Im}\lambda }$ is apparently finite
for any $t\geq 0$) 
\begin{equation}
z=z_{0}:=\frac{1}{2h}\left\vert \ln 2\int_{\Gamma }\left\vert
g_{0,t}(\lambda )\right\vert \frac{\left\vert d\lambda \right\vert }{\func{Im%
}\lambda }\right\vert .  \label{z0}
\end{equation}%
then 
\begin{equation*}
\left\vert \left\langle \mathbb{G}_{z,t}f,f\right\rangle \right\vert \leq 
\frac{1}{2}\left\Vert f\right\Vert_{L^2_{+}}^{2}
\end{equation*}%
and it now follows from $\left( \ref{ineq2}\right) $ that 
\begin{equation*}
\left\langle \left( I+\mathbb{G}_{z_{0},t}\right) f,f\right\rangle \geq 
\frac{1}{2}\left\Vert f\right\Vert_{L^2_{+}}^{2}
\end{equation*}%
which shows that for any $t\geq 0$ there is $z_{0}$ found by $\left( \ref{z0}%
\right) $ such that $I+\mathbb{G}_{z_{0},t}$ is invertible. Therefore by 
\cite{Steinberg69} $\left( I+\mathbb{G}_{z,t}\right) ^{-1}$ is a meromorphic
(operator) function of $z$ on the whole complex plane for any $t>0$.%
\footnote{%
for $t=0$ the operator $\mathbb{M}_{z,0}$ need not of course be analytic.}
The proposition is proven.
\end{proof}


\begin{lemma}
\label{Invert}Let $\mathbb{M}$ defined by $\left( \ref{K}\right) -\left( \ref%
{K2}\right) $ be compact. Then the operator $I+\mathbb{M}$ is boundedly
invertible if at least one of the following conditions holds:

\begin{enumerate}
\item $\mu \left( S\right) >0$ for some set $S$ of non-uniqueness of an $%
H_{+}^{2}$ function

\item $\left\vert \phi \left( \lambda \right) \right\vert <1$ a.e. on a set $%
S\subset \mathbb{R}$ of positive Lebesgue measure.
\end{enumerate}
\end{lemma}

\begin{proof}
The proof is standard. Since the operator $\mathbb{M}$ is compact, the point 
$-1$ may only be its eigenvalue. One needs to show that it is not the case.
Consider the homogeneous equation 
\begin{equation}
f+\mathbb{M}f=0.  \label{March}
\end{equation}%
Denoting $f_{+}=\chi_{+}f$, equation $\left( \ref{March}\right) $ therefore
implies%
\begin{equation*}
\left\langle f_{+},f_{+}\right\rangle +\left\langle \mathbb{M}%
f_{+},f_{+}\right\rangle =0
\end{equation*}%
or explicitly 
\begin{equation*}
\int \left\vert f_{+}\left( x\right) \right\vert ^{2}dx+\int_{\mathbb{R}%
_{+}}\left\vert \int e^{-\alpha x}f_{+}\left( x\right) dx\right\vert
^{2}d\mu \left( \alpha \right) +\int \phi \left( \lambda \right) \mathcal{F}%
f_{+}\left( \lambda \right) \mathcal{F}\overline{f_{+}}\left( \lambda
\right) d\lambda =0.
\end{equation*}%
Assuming that $\left\Vert f_{+}\right\Vert _{L^{2}}=1$ and observing that $%
\int e^{-\alpha x}f_{+}\left( x\right) dx=\sqrt{2\pi }\mathcal{F}f_{+}\left(
i\alpha \right) $ the last equation takes the form%
\begin{equation*}
1+2\pi \int_{\mathbb{R}_{+}}\left\vert \mathcal{F}f_{+}\left( i\alpha
\right) \right\vert ^{2}d\mu \left( \alpha \right) +\func{Re}\int \phi
\left( \lambda \right) \mathcal{F}f_{+}\left( \lambda \right) \mathcal{F}%
\overline{f_{+}}\left( \lambda \right) d\lambda =0.
\end{equation*}%
It follows from this equation that $\left(\left\Vert \mathcal{F}\overline{%
f_{+}}\right\Vert _{L^{2}}=1\right)$%
\begin{eqnarray*}
&&-2\pi \int_{\mathbb{R}_{+}}\left\vert \mathcal{F}f_{+}\left( i\alpha
\right) \right\vert ^{2}d\mu \left( \alpha \right) \\
&=&1+\func{Re}\int \phi \left( \lambda \right) \mathcal{F}f_{+}\left(
\lambda \right) \mathcal{F}\overline{f_{+}}\left( \lambda \right) d\lambda \\
&\geq &1-\left\vert \int \phi \left( \lambda \right) \mathcal{F}f_{+}\left(
\lambda \right) \mathcal{F}\overline{f_{+}}\left( \lambda \right) d\lambda
\right\vert \\
&\geq &1-\int \left\vert \phi \left( \lambda \right) \right\vert \left\vert 
\mathcal{F}f_{+}\left( \lambda \right) \right\vert \left\vert \mathcal{F}%
\overline{f_{+}}\left( \lambda \right) \right\vert d\lambda \\
&\geq &1-\left\Vert \phi \mathcal{F}f_{+}\right\Vert _{L^{2}}.
\end{eqnarray*}%
Thus 
\begin{equation}
\left\Vert \phi \mathcal{F}f_{+}\right\Vert _{L^{2}}\geq 1+2\pi \int_{%
\mathbb{R}_{+}}\left\vert \mathcal{F}f_{+}\left( i\alpha \right) \right\vert
^{2}d\mu \left( \alpha \right) .  \label{ineq}
\end{equation}%
If $\ \left\vert \phi \left( \lambda \right) \right\vert <1$ a.e. on a set $%
S $ of positive Lebesgue measure then $\left\Vert \phi \mathcal{F}%
f_{+}\right\Vert _{L^{2}}<\left\Vert \mathcal{F}f_{+}\right\Vert _{L^{2}}=1$
(as $\mathcal{F}f_{+}$ is in $H_{+}^{2}$ and hence cannot vanish on $S$) and 
$\left( \ref{ineq}\right) $ implies the obvious contradiction%
\begin{equation*}
0\leq \int_{\mathbb{R}_{+}}\left\vert \mathcal{F}f_{+}\left( i\alpha \right)
\right\vert ^{2}d\mu \left( \alpha \right) <0.
\end{equation*}%
Therefore $f_{+}=0$ and $I+\mathbb{M}$ is boundedly invertible. Assume now
that Condition 1 is satisfied. Without loss of generality we may assume $%
\left\vert \phi \left( \lambda \right) \right\vert =1$ a.e. on $\mathbb{R}$.
Then $\left\Vert \phi \mathcal{F}f_{+}\right\Vert _{L^{2}}=\left\Vert 
\mathcal{F}f_{+}\right\Vert _{L^{2}}=1$ and $\left( \ref{ineq}\right) $
implies 
\begin{equation*}
\int_{\mathbb{R}_{+}}\left\vert \mathcal{F}f_{+}\left( i\alpha \right)
\right\vert ^{2}d\mu \left( \alpha \right) \leq 0
\end{equation*}%
forcing $\mathcal{F}f_{+}\left( i\alpha \right) =0$ for every $\alpha \in S$%
. But $\mathcal{F}f_{+}$ is an $H_{+}^{2}$ function and hence cannot vanish
on $S$. The lemma is proven.
\end{proof}


\section{Main Results}

In this section we present our main results (given in two statements) which
will appear as simple consequences of the considerations above.


\subsection{The properties of the Marchenko operator}

The following statement relates the properties of a Marchenko type operator
with the properties of the underlying potential.


\begin{theorem}
\label{Structure}Let $V$ be a real function such that%
\begin{equation*}
\inf \limfunc{Spec}\left( -\partial _{x}^{2}+V\right) =-h_{0}>-\infty
\end{equation*}%
and that $V$ admits a decomposition%
\begin{equation}
V=V_{-}+V_{+}\ \left( V_{\pm }=\chi _{\pm }V\right)  \label{decomp}
\end{equation}%
where $V_{-}$ is arbitrary (subject to Hypothesis \ref{H1} at $-\infty $)
and $V_{+}\in L_{+}^{1}\left( \left\langle x\right\rangle dx\right) $.
Consider 
a two parametric family of Marchenko type (Definition \ref{March type})
operators $\mathbb{M}_{z,t}$ ($z\in \mathbb{C}$ and $t\geq 0$) associated
with the data $\left( \rho _{z,t},R_{z,t}\right) $ where $d\rho _{z,t}\left(
\alpha \right) :=e^{-2\lambda z+8\alpha ^{3}t}d\rho \left( \alpha \right) $,
and $d\rho \left( \alpha \right) $ is defined by%
\begin{equation}
d\rho \left( \alpha \right) =\left\{ \dsum_{n=1}^{N}\left( c_{n}^{+}\right)
^{2}\delta \left( \alpha -\kappa _{n}^{+}\right) +\left( \dfrac{T_{+}^{2}%
\func{Im}R_{-}}{\left\vert 1-R_{-}L_{+}\right\vert ^{2}}\right) \left(
+0+i\alpha \right) \right\} d\alpha ,  \label{d_rou}
\end{equation}%
and 
\begin{equation*}
R_{z,t}\left( \lambda \right) :=e^{i(2\lambda z+8\lambda ^{3}t)}R\left(
\lambda \right) .
\end{equation*}%
The measure $\rho $ is non-negative, finite, supported on $\left[ 0,h_{0}%
\right] $ and independent of the choice of the splitting point in $\left( %
\ref{decomp}\right) $. The operator $\mathbb{M}_{z,t}$ therefore is well
defined and has the following properties:

\begin{enumerate}
\item \label{Prop1}$\mathbb{M}_{z,t}$ is selfadjoint and bounded for any
real $z$ and $t\geq 0$.

\item \label{Prop2}If $V_{+}\in L_{+}^{2}\left( e^{\delta x^{1/2}}dx\right) $
for some $\delta >0$ then $\mathbb{M}_{z,t}$ is a trace type real analytic
operator-valued function in $z$ for any $t>0$ and $\left( I+\mathbb{M}%
_{z,t}\right) ^{-1}$ is real meromorphic\footnote{%
I.e. a ratio of two real analytic functions.
\par
{}}.

\item \label{Prop3}If $V_{+}=0$ then $\mathbb{M}_{z,t}$ is an entire in $z$
operator-valued function of trace class for any $t>0$ and $\left( I+\mathbb{M%
}_{z,t}\right) ^{-1}$ is meromorphic in $z$.

\item \label{Prop4}If $\limfunc{Spec}_{\limfunc{ac}}\left( -\partial
_{x}^{2}+V\right) $ has a nonempty component of multiplicity two then $I+%
\mathbb{M}_{z,t}$ is boundedly invertible for any real $z$ and $t\geq 0.$
\end{enumerate}
\end{theorem}

\begin{proof}
Prove first that $\rho \left( \alpha \right) $ defined by $\left( \ref{d_rou}%
\right) $ is independent of the particular decomposition $\left( \ref{decomp}%
\right) $. Take $\widetilde{V}$ from the proof of Lemma \ref{Fragmentation}.
Since $\widetilde{V}\in L^{1}\left( \left\langle x\right\rangle dx\right) $
the measure $\widetilde{\rho }\left( \alpha \right) $ can be alternatively
computed by%
\begin{equation*}
d\widetilde{\rho }\left( \alpha \right) =\dsum_{n=1}^{\widetilde{N}}{%
\widetilde{c}_{n}}^{\ 2}\delta \left( \alpha -\widetilde{\kappa }_{n}\right)
d\alpha
\end{equation*}%
and hence it is independent of the splitting point in $\left( \ref{decomp}%
\right) $. On the other hand, 
\begin{equation}
d\widetilde{\rho }\left( \alpha \right) =\left\{ \dsum_{n=1}^{N}\left(
c_{n}^{+}\right) ^{2}\delta \left( \alpha -\kappa _{n}^{+}\right) +\left( 
\dfrac{T_{+}^{2}\func{Im}\widetilde{R}_{-}}{\left\vert 1-\widetilde{R}%
_{-}L_{+}\right\vert ^{2}}\right) \left( i\alpha +0+\right) \right\} d\alpha
.  \label{d_rou_tilda}
\end{equation}%
Note that $T_{+}$ and $L_{+}$ are real on the imaginary line 
\begin{eqnarray}
\left( \dfrac{T_{+}^{2}\func{Im}\widetilde{R}_{-}}{\left\vert 1-\widetilde{R}%
_{-}L_{+}\right\vert ^{2}}\right) \left( +0+i\alpha \right)
&=&T_{+}^{2}\left( +0+i\alpha \right) \func{Im}\left( \dfrac{\widetilde{R}%
_{-}-\left\vert \widetilde{R}_{-}\right\vert ^{2}\overline{L_{+}}}{%
\left\vert 1-\widetilde{R}_{-}L_{+}\right\vert ^{2}}\right) \left(
+0+i\alpha \right)  \notag \\
&=&T_{+}^{2}\left( +0+i\alpha \right) \func{Im}\left( \dfrac{\widetilde{R}%
_{-}}{1-\widetilde{R}_{-}L_{+}}\right) \left( +0+i\alpha \right)  \notag \\
&=&\func{Im}\left( \dfrac{\widetilde{R}_{-}}{1-\widetilde{R}_{-}L_{+}}%
T_{+}^{2}\right) \left( +0+i\alpha \right)  \notag \\
&=&\func{Im}\widetilde{G}\left( +0+i\alpha \right)  \label{x}
\end{eqnarray}%
where $\widetilde{G}$ is defined by $\left( \ref{Rtilda1}\right) $.
Inserting $\left( \ref{x}\right) $ into $\left( \ref{d_rou_tilda}\right) $
we have 
\begin{equation*}
d\widetilde{\rho }\left( \alpha \right) =\left\{ \dsum_{n=1}^{N}\left(
c_{n}^{+}\right) ^{2}\delta \left( \alpha -\kappa _{n}^{+}\right) +\func{Im}%
\widetilde{G}\left( +0+i\alpha \right) \right\} d\alpha .
\end{equation*}%
As proven in Lemma \ref{Fragmentation}, 
\begin{eqnarray*}
\lim_{b\rightarrow \infty }\func{Im}\widetilde{G}\left( +0+i\alpha \right)
d\alpha &=&\func{Im}G\left( +0+i\alpha +0\right) d\alpha \\
&=&\left( \dfrac{T_{+}^{2}\func{Im}R_{-}}{\left\vert 1-R_{-}L_{+}\right\vert
^{2}}\right) \left( +0+i\alpha \right) d\alpha
\end{eqnarray*}%
and therefore $\lim_{b\rightarrow \infty }d\widetilde{\rho }\left( \alpha
\right) =d\rho \left( \alpha \right) $. But each $\widetilde{\rho }$ is
independent of the split in $\left( \ref{decomp}\right) $ and thus $\rho $
doesn't depend on $\left( \ref{decomp}\right) $. We prove now that $\rho
\left( \alpha \right) $ is a non-negative finite measure on $\left[ 0,h_{0}%
\right] $. Non-negativity follows from $\left( \ref{x}\right) $ (with the
tilde dropped and using \eqref{R})%
\begin{align*}
\func{Im}G\left( +0+i\alpha \right) & =\func{Im}\left( \dfrac{R_{-}}{%
1-R_{-}L_{+}}T_{+}^{2}\right) \left( +0+i\alpha \right) \\
& =2\alpha \left( \dfrac{T_{+}^{2}}{\left\vert 1-R_{-}L_{+}\right\vert
^{2}\left\vert i\lambda +m_{-}\right\vert ^{2}}\right) \left( +0+i\alpha
\right) \func{Im}m_{-}\left( -\alpha ^{2}+i0+\right) \\
& \geq 0.
\end{align*}%
%
%
%
%
%
%
%
%
%
We show now that $\rho \left( \alpha \right) $ is a finite measure on $\left[
0,h_{0}\right] $. Choose, for simplicity, the splitting point so that $%
-\partial _{x}^{2}+V_{-}$ has only one bound state $-\kappa _{0}^{2}$ and
evaluate $\left( \ref{d_rou}\right) $ around $\alpha =\kappa _{0}\in \left[
0,h_{0}\right] $ and $\alpha =0$ separately. From $\left( \ref{R}\right) $
and \eqref{L+} one has 
\begin{eqnarray}
G &=&\dfrac{R_{-}}{1-R_{-}L_{+}}T_{+}^{2}  \notag \\
&=&\left( 1-\frac{i\lambda -m_{-}}{i\lambda +m_{-}}\frac{i\lambda -m_{+}}{%
i\lambda +m_{+}}\right) ^{-1}\frac{i\lambda -m_{-}}{i\lambda +m_{-}}T_{+}^{2}
\notag \\
&=&\left( -1+\frac{i\lambda +m_{+}}{m_{-}+m_{+}}\right) \frac{i\lambda +m_{+}%
}{2i\lambda }T_{+}^{2}  \notag \\
&=&\left( -1+\frac{i\lambda +m_{+}}{m_{-}+m_{+}}\right) \frac{2i\lambda }{%
i\lambda +m_{+}}g  \notag \\
&=&\frac{2i\lambda }{m_{-}+m_{+}}g-\frac{2i\lambda }{i\lambda +m_{+}}g,
\label{G}
\end{eqnarray}%
where 
\begin{equation*}
g:=\left( \frac{i\lambda +m_{+}}{2i\lambda }T_{+}\right) ^{2}.
\end{equation*}%
It follows from \eqref{L+} 
that 
\begin{equation}
\frac{i\lambda +m_{+}\left( \lambda ^{2}\right) }{2i\lambda }=\left(
1+L_{+}\left( \lambda \right) \right) ^{-1}  \label{eq8.6}
\end{equation}%
and hence 
\begin{equation*}
g\left( i\alpha \right) =\left( \frac{T_{+}\left( i\alpha \right) }{%
1+L_{+}\left( i\alpha \right) }\right) ^{2}.
\end{equation*}%
Since $T_{+}$ and $L_{+}$ both have a simple pole at $\lambda =i\kappa _{0}$%
, the function $g$ has a removable singularity at $\lambda =i\kappa _{0}$.
Moreover, due to the symmetry, $T_{+}\left( i\alpha \right) $ and $%
L_{+}\left( i\alpha \right) $ are both real and hence $g\left( i\alpha
\right) \geq 0$ and bounded away from $\alpha =0$.

Note that since $m_{\pm }\left( z\right) $ and $i\sqrt{z}$ are both
Herglotz, i.e. $\mathbb{C}_{+}\rightarrow \mathbb{C}_{+}$, we immediately
conclude that $-\left( m_{-}\left( z\right) +m_{+}\left( z\right) \right)
^{-1}$ and $-\left( m_{+}\left( z\right) +i\sqrt{z}\right) ^{-1}$ are also
Herglotz and hence admit a Herglotz representation similar to $\left( \ref%
{Herglotz3.3}\right) $ with some non-negative finite measures $\mu _{1}$ and 
$\mu _{2}$ respectively computed by%
\begin{align*}
d\mu _{1}\left( s\right) &=-\frac{1}{\pi }\func{Im}\left( m_{-}\left(
s+i0+\right) +m_{+}\left( s+i0+\right) \right) ^{-1}ds. \\
d\mu _{2}\left( s\right) &=-\frac{1}{\pi }\func{Im}\left( m_{+}\left(
s+i0+\right) +i\sqrt{s}\right) ^{-1}ds.
\end{align*}%
Therefore it follows from $\left( \ref{drou}\right) $ and $\left( \ref{G}%
\right) $ that%
\begin{eqnarray*}
d\rho \left( \alpha \right) &=&\frac{1}{\pi }g\left( i\alpha \right) \left\{
-\func{Im}\left( m_{-}\left( -\alpha ^{2}+i0+\right) +m_{+}\left( -\alpha
^{2}+i0+\right) \right) ^{-1}\left( -2\alpha \right) \right. \\
&&\left. +\func{Im}\left( m_{+}\left( -\alpha ^{2}+i0+ \right) -\alpha
\right) ^{-1}\left( -2\alpha \right) \right\} d\alpha \\
&=&g\left( i\alpha \right) \left( d\mu _{1}\left( -\alpha ^{2}\right) -d\mu
_{2}\left( -\alpha ^{2}\right) \right) \\
&=&g\left( i\alpha \right) \left( -d\rho _{1}\left( \alpha \right) +d\rho
_{2}\left( \alpha \right) \right) ,
\end{eqnarray*}%
where $d\rho _{k}\left( \alpha \right) :=-d\mu _{k}\left( -\alpha
^{2}\right) $ $k=1,2$ and finite (non-negative) measures. Since, as already
proven, $g\left( i\alpha \right) $ is bounded away from $\alpha=0$ the
measure $\rho \left( \alpha \right) $ is real and finite on $\left[
\varepsilon ,h_{0}\right] $, $\varepsilon >0$. If $\alpha=0$ is an
exceptional point then $\lim\limits_{\alpha \rightarrow +0}L_{+}\left(
i\alpha \right) >-1$ and $\lim\limits_{\alpha \rightarrow +0}g\left( i\alpha
\right) $ is finite and $\rho \left( \alpha \right) $ is non-negative and
finite on $\left[ 0,h_{0}\right] $. It remains to show that $\rho \left(
\alpha \right) $ is finite on $\left[ 0,h_{0}\right] $ even if $%
\lim\limits_{\alpha \rightarrow +0}L_{+}\left( i\alpha \right) =-1$, i.e. $%
\alpha =0$ is a generic point. To this end we need to represent $G$
differently:%
\begin{eqnarray*}
G &=&2i\lambda \frac{\left( i\lambda -m_{-}\right) \left( i\lambda
+m_{+}\right) }{m_{-}+m_{+}}\left( \frac{T_{+}}{2i\lambda }\right) ^{2} \\
&=&2i\lambda \left\{ \frac{\left( i\lambda \right) ^{2}}{m_{-}+m_{+}}-\frac{%
m_{-}m_{+}}{m_{-}+m_{+}}-i\lambda +\frac{2i\lambda m_{+}}{m_{-}+m_{+}}%
\right\} \omega \\
&=&G_{1}+G_{2}+G_{3}+G_{4}\ ,\ \ \ \ \ \omega :=\left( \frac{T_{+}}{%
2i\lambda }\right) ^{2}.
\end{eqnarray*}%
If $\alpha =0$ is a generic point $\dfrac{T_{+}\left( i\alpha \right) }{%
\alpha }$ remains bounded as $\alpha \rightarrow +0$ and hence so does $%
\omega \left( i\alpha \right) $.

The term $G_{3}$ is trivial:%
\begin{equation*}
\frac{1}{\pi }\func{Im}G_{3}\left( +0+i\alpha \right) d\alpha =-\frac{2}{\pi 
}\alpha ^{2}\omega \left( i\alpha \right) d\alpha .
\end{equation*}%
Since $-\left( m_{-}+m_{+}\right) ^{-1}$ and $\dfrac{m_{-}m_{+}}{m_{-}+m_{+}}%
=\dfrac{1}{1/m_{-}+1/m_{+}}$ are Herglotz, by the same arguments as above,
one can easily conclude that the measures 
\begin{equation*}
\dfrac{1}{\pi }\func{Im}G_{k}\left( +0+i\alpha \right) d\alpha ,:k=1,2
\end{equation*}%
\ are finite on $\left[ 0,\varepsilon \right] $. The measure produced by 
\begin{equation*}
G_{4}=\frac{\left( 2i\lambda \right) ^{2}m_{+}}{m_{-}+m_{+}}\omega =\frac{%
2i\lambda }{m_{-}+m_{+}}\cdot 2i\lambda m_{+}\omega
\end{equation*}%
requires a bit more care. Since the factor$\dfrac{2i\lambda }{m_{-}+m_{+}}$
has already been analyzed above, one needs to make sure that $\alpha
m_{+}\left( -\alpha ^{2}\right) \omega \left( i\alpha \right) $ stays
bounded as $\alpha \rightarrow 0$. From $\left( \ref{eq8.6}\right) $ 
\begin{equation}
-\alpha m_{+}\left( -\alpha ^{2}\right) =-\alpha ^{2}+\frac{2\alpha ^{2}}{%
1+L_{+}\left( i\alpha \right) }.  \label{1'}
\end{equation}%
It follows from $\left( \ref{L}\right) $ and $\left( \ref{T}\right) $ that 
\begin{eqnarray*}
1+L_{+}\left( i\alpha \right) &=&1-\frac{\int_{0}^{\infty }e^{-2\alpha
x}V\left( x\right) y_{+}\left( x,i\alpha \right) dx}{2\alpha
+\int_{0}^{\infty }V\left( x\right) y_{+}\left( x,i\alpha \right) dx} \\
&=&\frac{2\alpha +\int_{0}^{\infty }V\left( x\right) \left( 1-e^{-2\alpha
x}\right) y_{+}\left( x,i\alpha \right) dx}{2\alpha +\int_{0}^{\infty
}V\left( x\right) y_{+}\left( x,i\alpha \right) dx} \\
&=&2\alpha \frac{1+\int_{0}^{\infty }V\left( x\right) \dfrac{1-e^{-2\alpha x}%
}{2\alpha }y_{+}\left( x,i\alpha \right) dx}{2\alpha +\int_{0}^{\infty
}V\left( x\right) y_{+}\left( x,i\alpha \right) dx}
\end{eqnarray*}%
where $y_{+}\left( x,i\alpha \right) $ solves the integral equation $\left( %
\ref{y+-}\right) $ 
\begin{equation}
y_{+}\left( x,i\alpha \right) =1+\int_{x}^{\infty }\frac{1-e^{-2\alpha (s-x)}%
}{2\alpha }V\left( s\right) y_{+}\left( s,i\alpha \right) ds.  \label{y+}
\end{equation}%
Hence for the second term on the right hand side\ of $\left( \ref{1'}\right) 
$ we have%
\begin{align}
\frac{2\alpha ^{2}}{1+L_{+}\left( i\alpha \right) }& =\dfrac{\alpha \left(
2\alpha +\dint_{0}^{\infty }V\left( x\right) y_{+}\left( x,i\alpha \right)
dx\right) }{1+\dint_{0}^{\infty }\dfrac{1-e^{-2\alpha x}}{2\alpha }V\left(
x\right) y_{+}\left( x,i\alpha \right) dx}  \label{ineq3} \\
& \leq \dfrac{\alpha \left( 2\alpha +\dint_{0}^{\infty }\left\vert V\left(
x\right) \right\vert \left\vert y_{+}\left( x,i\alpha \right) \right\vert
dx\right) }{\left\vert 1-\dint_{0}^{\infty }\dfrac{1-e^{-2\alpha x}}{2\alpha 
}\left\vert V\left( x\right) \right\vert \left\vert y_{+}\left( x,i\alpha
\right) \right\vert dx\right\vert }.  \notag
\end{align}%
From $\left( \ref{y+}\right) $ one has%
\begin{eqnarray*}
\left\vert y_{+}\left( x,i\alpha \right) \right\vert &\leq
&1+\int_{x}^{\infty }\frac{1-e^{-2\alpha (s-x)}}{2\alpha }\left\vert V\left(
s\right) \right\vert \left\vert y_{+}\left( s,i\alpha \right) \right\vert ds
\\
&\leq &1+\int_{x}^{\infty }\left( s-x\right) \left\vert V\left( s\right)
\right\vert \left\vert y_{+}\left( s,i\alpha \right) \right\vert ds \\
&\leq &1+\int_{x}^{\infty }s\left\vert V\left( s\right) \right\vert
\left\vert y_{+}\left( s,i\alpha \right) \right\vert ds.
\end{eqnarray*}%
Iterating this inequality immediately produces%
\begin{eqnarray*}
\left\vert y_{+}\left( x,i\alpha \right) \right\vert &\leq
&\dsum\limits_{n\geq 0}\left( \int_{0}^{\infty }x\left\vert V\left( x\right)
\right\vert dx\right) ^{n} \\
&\leq &\dsum\limits_{n\geq 0}\left\Vert V\right\Vert _{L_{+}^{1}\left(
\left\langle x\right\rangle dx\right) }^{n}=\frac{1}{1-\left\Vert
V\right\Vert _{L_{+}^{1}\left( \left\langle x\right\rangle dx\right) }}.
\end{eqnarray*}%
Denoting $\varepsilon =\left\Vert V\right\Vert _{L_{+}^{1}\left(
\left\langle x\right\rangle dx\right) }$ and taking it small enough, the
inequality $\left( \ref{ineq3}\right) $ then yields%
\begin{eqnarray*}
\frac{2\alpha ^{2}}{1+L_{+}\left( i\alpha \right) } &\leq &\frac{\alpha }{1-%
\dfrac{1}{1-\varepsilon }\dint_{0}^{\infty }x\left\vert V\left( x\right)
\right\vert dx}\left( 2\alpha +\frac{1}{1-\varepsilon }\int_{0}^{\infty
}\left\vert V\left( x\right) \right\vert dx\right) \\
&\leq &\frac{2\alpha +\varepsilon }{1-2\varepsilon }\alpha \rightarrow
0,:\alpha \rightarrow 0.
\end{eqnarray*}%
Thus the measure $d\rho $ is finite.

We now prove the Properties (\ref{Prop1})-(\ref{Prop4}). We start by
splitting 
\begin{equation}
\mathbb{M}_{z,t}=\mathbb{M}_{z,t}^{+}+\mathbb{G}_{z,t}  \label{M_split}
\end{equation}%
where $\mathbb{M}_{z,t}^{+}$ is the Marchenko type operator introduced in
Definition \ref{March type} with%
\begin{eqnarray*}
d\rho \left( \alpha \right) &=&e^{-2\alpha z+8\alpha
^{3}t}\dsum_{n=1}^{N}\left( c_{n}^{+}\right) ^{2}\delta \left( \alpha
-\kappa _{n}^{+}\right) d\alpha , \\
\phi \left( \lambda \right) &=&e^{2i\lambda x+8i\lambda ^{3}t}R_{+}(\lambda
),
\end{eqnarray*}%
and $\mathbb{G}_{z,t}$ defined by $\left( \ref{G-oper}\right) -\left( \ref{g}%
\right) $ with $G=\dfrac{R_{-}}{1-R_{-}L_{+}}T_{+}^{2}$. Let us show that $G$
satisfies all the conditions of Proposition \ref{Properties of G}. Indeed
each function $R_{-},L_{+},T_{+}$ is clearly subject to conditions (i)-(iv)
of Proposition \ref{Properties of G}. The function $G$ is then immediately
subject to conditions (i)-(ii) of Proposition \ref{Properties of G}. The
existence of boundary values of $G$ on the real line is also obvious. It is
also in $L^{\infty }$ since 
\begin{equation*}
\left\vert G\left( \lambda +i0\right) \right\vert =\left\vert R\left(
\lambda \right) -R_{+}\left( \lambda \right) \right\vert \leq 2.
\end{equation*}%
Thus $G$ satisfies condition (iii) of Proposition \ref{Properties of G}.
Condition (iv) was verified above. By Proposition \ref{Properties of G} (\ref%
{pr7.2.2}) $\mathbb{G}_{z,t}$ is then selfadjoint for real $z$'s for any $%
t>0 $. As was shown above $\mathbb{M}_{z,t}^{+}$ is Marchenko operator
corresponding to $V_{+}$ and its selfadjointness for real $z$'s is a
well-known fact. Thus $\mathbb{M}_{z,t}$ is selfadjoint for real $z$ and
positive $t$. The boundedness of $\mathbb{M}_{z,t}$ follows from the
finiteness of $d\rho $ (see, e.g. \cite{Peller2003}). This proves Property %
\ref{Prop1}. Turn now to Property \ref{Prop2}. Under condition $V_{+}\in
L_{+}^{2}\left( e^{\delta x^{1/2}}dx\right) $ for some $\delta >0$, the
operator $\mathbb{M}_{z,t}^{+}$ is real analytic in $z$ for any $t>0$ \cite%
{Tarama04} and by Proposition \ref{Properties of G} (\ref{pr7.2.4}) $\mathbb{%
G}_{z,t}$ is entire. Therefore, $\mathbb{M}_{z,t}$ is real analytic and \cite%
{Steinberg69} $\left( I+\mathbb{M}_{z,t}\right) ^{-1}$ is real meromorphic
in $z$ for any $t>0$. By Lemma \ref{Trace}, $\mathbb{G}_{z,t}\in \mathfrak{S}%
_{1}$ and by Lemma \ref{Trace type}, $\mathbb{M}_{z,t}^{+}$ is a trace type
operator. Thus Property \ref{Prop2} is proven. Note that if $V_{+}=0$ then $%
\mathbb{M}_{z,t}^{+}=0$ and Property \ref{Prop3} immediately follows.

Assume now that $-\partial _{x}^{2}+V$ has a nontrivial a.c. component $S$
of multiplicity 2 and hence, by Lemma \ref{Fragmentation}, $\left\vert
R\left( \lambda \right) \right\vert <1$ a.e. on $S$. By Lemma \ref{Invert}, $%
I+\mathbb{M}_{z,t}$ is boundedly invertible for any real $z$ and $t\geq 0$
and Property \ref{Prop4} is proven.
\end{proof}


\begin{remark}
If $V\left( x\right) =-h^{2}\chi _{-}\left( x\right) $ then 
\begin{equation*}
R\left( \lambda \right) =-\left( \frac{h}{\lambda +\sqrt{\lambda ^{2}+h^{2}}}%
\right) ^{2}
\end{equation*}
and $\rho $ is absolutely continuous, supported on $\left[ 0,h\right] $ and 
\begin{equation*}
d\rho \left( \alpha \right) =\frac{2}{\pi h^{2}}\alpha \sqrt{h^{2}-\alpha
^{2}}d\alpha .
\end{equation*}
\end{remark}


\begin{remark}
\label{Known cases}Marchenko operators for the cases of $V$'s such that $%
V_{-}\left( x\right) -p\left( x\right) \in L_{-}^{1}\left( \left\langle
x\right\rangle dx\right) $ with either $p\left( x\right) =\limfunc{const}$
or periodic and $V_{+}\left( x\right) \in L_{+}^{1}\left( \left\langle
x\right\rangle dx\right) $ (so-called step-like potentials) have been
considered by many authors in the connection with inverse problems (see,
e.g., \cite{AK01}) and IST for KdV ( \cite{Hruslov76}, \cite{Cohen1984}, 
\cite{Venak86}, \cite{KK94}, \cite{Teschl2}, etc.), scattering quantities
being typically introduced differently from ours.
\end{remark}


\begin{remark}
If $V_{+}\in L_{+}^{1}\left( e^{\delta x}dx\right) $ for some $\delta >0$
then $\mathbb{M}_{z,t}$ is entire and $\left( I+\mathbb{M}_{z,t}\right)
^{-1} $ is meromorphic in $z$ on $\mathbb{C}$ for any $t>0.$
\end{remark}


\begin{remark}
Loosely speaking, the measure $\rho $ carries over the information about the
negative spectrum and $R$ does it for the positive spectrum but they need
not be independent.
\end{remark}


\subsection{Determinant solution to the Cauchy problem for the KdV equation}


Our main result is given in the following theorem.

\begin{theorem}
\label{MainThm}Let real $V_{0}$ be such that $\limfunc{Spec}\left( -\partial
_{x}^{2}+V_{0}\right) $ is bounded from below and has a non-empty twofold
a.c. spectrum. Assume that%
\begin{eqnarray}
\chi _{-}V_{0} &\in &L_{-}^{2}\left( e^{-\delta _{-}\left\vert x\right\vert
}dx\right)  \label{decay at -infty} \\
\chi _{+}V_{0} &\in &L_{+}^{2}\left( e^{\delta _{+}x^{1/2}}dx\right)
\label{decay at +fty}
\end{eqnarray}%
for some $\delta _{\pm }>0$. Then the Cauchy problem for the KdV equation%
\begin{equation}
\left\{ 
\begin{array}{c}
\partial _{t}V-6V\partial _{x}V+\partial _{x}^{3}V=0 \\ 
V\left( x,0\right) =V_{0}\left( x\right)%
\end{array}%
\right.  \label{kdv}
\end{equation}%
has a unique global natural solution $V\left( x,t\right) $ given by 
\begin{equation}
V\left( x,t\right) =-2\partial _{x}^{2}\log \limfunc{Det}\left( I+\mathbb{M}%
_{x,t}\right)  \label{Det}
\end{equation}%
where $\mathbb{M}_{x,t}$ is defined in Theorem \ref{Structure}, $V\left(
x,t\right) $ being a real analytic function in $x$ for any $t>0$. Moreover,
for any $a>-\infty $ 
\begin{equation}
\lim_{t\rightarrow +0}\left\Vert V\left( \cdot ,t\right) -V_{0}\right\Vert
_{L^{2}\left( a,\infty \right) }=0.  \label{IC}
\end{equation}
\end{theorem}

\begin{proof}
Let $V_{0,n}\left( x\right) $ be an arbitrary sequence of real compactly
supported $L^{2}$ functions such $\limfunc{Supp}V_{0,n}=\left(
a_{n},b_{n}\right) $ and 
\begin{equation*}
\left\Vert V_{0,n}-V_{0}\right\Vert _{L_{\limfunc{loc}}^{2}}\rightarrow
0,:n\rightarrow \infty .
\end{equation*}%
Then the KdV equation with initial data $V_{0,n}\left( x\right) $ has a
unique solution $V_{n}\left( x,t\right) $ computed by the standard inverse
scattering transform%
\begin{equation*}
V_{n}\left( x,t\right) =-2\partial _{x}^{2}\log \limfunc{Det}\left( I+%
\mathbb{M}_{n,x,t}\right) ,
\end{equation*}%
where $\mathbb{M}_{n,x,t}$ is the Marchenko operator corresponding to $%
V_{0,n}$. By Theorem \ref{Structure}, each $V_{n}\left( x,t\right) $ is a
meromorphic function in $x$ on the entire complex plane. Consider the
function $V\left( x,t\right) $\ given by $\left( \ref{Det}\right) $. By
Theorem \ref{Structure}, it is well defined and real analytic in $x$ for any 
$t>0$ and it remains to prove that it solves $\left( \ref{kdv}\right) $. To
this end we rewrite $V=V_{n}+\Delta V_{n}$, where $\Delta V_{n}:=V-V_{n}$,
and insert this into the left hand side of $\left( \ref{kdv}\right) $:%
\begin{equation*}
\partial _{t}V-6V\partial _{x}V+\partial _{x}^{3}V \\
=\partial _{t}\Delta V_{n}+3\partial _{x}\left[ \left( \Delta
V_{n}-2V\right) \Delta V_{n}\right] +\partial _{x}^{3}\Delta V_{n}.
\end{equation*}%
Using $\left( \ref{M_split}\right) $ by a straightforward computation we
have (dropping subscipts $x,t$)%
\begin{align}
\Delta V_{n}\left( x,t\right) & =\Delta V_{n}^{+}\left( x,t\right)  \notag \\
& \quad +2\partial _{x}^{2}\log \limfunc{Det}\left( I+\left( I+\mathbb{M}%
_{n}^{+}\right) ^{-1}\mathbb{G}_{n}\right) \left( I+\left( I+\mathbb{M}%
^{+}\right) ^{-1}\mathbb{G}\right) ^{-1}  \notag \\
& =\Delta V_{n}^{+}\left( x,t\right)  \notag \\
& \quad +2\partial _{x}^{2}\log \det \left\{ I-\left( I+\mathbb{M}\right)
^{-1}\left( I+\mathbb{M}^{+}\right) \left( I+\mathbb{M}_{n}^{+}\right)
^{-1}\Delta \Omega \right\}  \label{delta V}
\end{align}%
with 
\begin{equation*}
\Delta \Omega =\Delta \mathbb{G}_{n}-\Delta \mathbb{M}_{n}^{+}\left( I+%
\mathbb{M}^{+}\right) ^{-1}\mathbb{G},
\end{equation*}%
where the determinant on the right hand side of \eqref{delta V} is
understood in the usual way because by Lemma \ref{Trace} both $\Delta 
\mathbb{G}$ and $\mathbb{G}$ are trace class, and $V_{n}^{+}$ stands for the
solution to the KdV equation with the initial profile $\chi _{+}V_{0,n}$.
Since $V_{n}^{+}\left( x,t\right) \rightarrow V^{+}\left( x,t\right) $
uniformly in $x$ as $n\rightarrow \infty $ (one of the main results of \cite%
{Tarama04}), the first term on the right hand side of $\left( \ref{delta V}%
\right) $ vanishes as $n\rightarrow \infty $. We now need to show that so
does the other one. But by \cite{Tarama04} $\left\Vert \Delta \mathbb{M}%
_{n}^{+}\right\Vert \rightarrow 0,n\rightarrow \infty ,$ and therefore 
\begin{equation}
\left\Vert \Delta \mathbb{M}_{n}^{+}\left( I+\mathbb{M}^{+}\right) ^{-1}%
\mathbb{G}\right\Vert _{\mathfrak{S}_{1}}\leq \left\Vert \left( I+\mathbb{M}%
^{+}\right) ^{-1}\right\Vert \left\Vert \mathbb{G}\right\Vert _{\mathfrak{S}%
_{1}}\left\Vert \Delta \mathbb{M}_{n}^{+}\right\Vert \rightarrow
0,:n\rightarrow \infty ,  \label{ineq1}
\end{equation}%
\begin{align}
\left\Vert \left( I+\mathbb{M}_{n}^{+}\right) ^{-1}\right\Vert & \leq
\left\Vert \left( I+\mathbb{M}^{+}\right) ^{-1}\right\Vert \left\Vert \left(
I-\left( I+\mathbb{M}^{+}\right) ^{-1}\Delta \mathbb{M}_{n}^{+}\right)
^{-1}\right\Vert  \label{ineq1'} \\
& \leq \left\Vert \left( I+\mathbb{M}^{+}\right) ^{-1}\right\Vert \left(
1-\left\Vert \left( I+\mathbb{M}^{+}\right) ^{-1}\right\Vert \left\Vert
\Delta \mathbb{M}_{n}^{+}\right\Vert \right) ^{-1}  \notag \\
& \quad \rightarrow \left\Vert \left( I+\mathbb{M}^{+}\right)
^{-1}\right\Vert ,:n\rightarrow \infty .  \notag
\end{align}%
Next, by Lemma \ref{Trace} (reinstating subscripts $x,t$)%
\begin{eqnarray}
\left\Vert \Delta \mathbb{G}_{n,x,t}\right\Vert _{\mathfrak{S}_{1}}
&\lesssim &\left\Vert \frac{\Delta {g}_{n,x.t}\left( \lambda \right) }{\func{%
Im}\lambda }\right\Vert _{L^{1}\left( \Gamma \right) }=\left\Vert \frac{%
e^{2i\lambda x}e^{8i\lambda ^{3}t}\Delta G_{n}(\lambda )}{\func{Im}\lambda }%
\right\Vert _{L^{1}\left( \Gamma \right) }  \notag \\
&\leq &\left\Vert \frac{e^{2i\lambda x}e^{8i\lambda ^{3}t}\Delta
G_{n}(\lambda )}{\func{Im}\lambda }\right\Vert _{L^{1}\left( \Gamma
_{N}\right) }+\left\Vert \frac{e^{2i\lambda x}e^{8i\lambda ^{3}t}\Delta
G_{n}(\lambda )}{\func{Im}\lambda }\right\Vert _{L^{1}\left( \Gamma
\backslash \Gamma _{N}\right) }  \label{L1 norm}
\end{eqnarray}%
where $\Gamma $ is as in Figure \ref{fig1} and $\Gamma _{N}:=\Gamma \cap
\left\{ \lambda :\left\vert \lambda \right\vert \leq N\right\} $. We need to
show that each term on the right hand side of (\ref{L1 norm}) is small for
large $n$ and $N$. For $\Delta G_{n}$ we have 
\begin{eqnarray*}
\Delta G_{n} &=&\Delta \dfrac{R_{n,-}}{1-R_{n,-}L_{n,+}}T_{n,+}^{2} \\
&=&\Delta \left( \dfrac{T_{n,+}^{2}}{L_{n,+}}\right) \left( \dfrac{1}{%
1-R_{-}L_{+}}-1\right) +\dfrac{T_{n,+}^{2}}{L_{n,+}}\dfrac{L_{n,+}\Delta
R_{n,-}+R_{-}\Delta L_{n,+}}{\left( 1-R_{n,-}L_{n,+}\right) \left(
1-R_{-}L_{+}\right) }.
\end{eqnarray*}%
But \cite{CarLac90}%
\begin{equation*}
\left\Vert V_{0,n}-V_{0}\right\Vert _{L_{\limfunc{loc}}^{2}}\rightarrow
0\quad \Longrightarrow \quad m_{\pm ,n,0}\left( z\right) \rightarrow m_{\pm
,0}\left( z\right) ,\ \ \ n\rightarrow \infty ,
\end{equation*}%
the latter convergence being uniform in $z$ on compacts in $\mathbb{C}_{+}$.
Due to (\ref{R}) and (\ref{L+}) then $\Delta R_{n,-}\rightarrow 0,:\Delta
L_{n,+}\rightarrow 0,\ \ n\rightarrow \infty ,$ also uniformly on compacts
in $\mathbb{C}_{+}$. Since by \cite{Tarama04} we also have $\Delta
T_{n,+}\rightarrow 0,\ \ n\rightarrow \infty $, we conclude that $\Delta
G_{n}\rightarrow 0,\ \ n\rightarrow \infty $ uniformly on $\Gamma _{N}$ and
hence the first norm on the right hand side of $\left( \ref{L1 norm}\right) $
is small for $n$ large enough. The second norm on the right hand side of $%
\left( \ref{L1 norm}\right) $ is small if $N$ is large enough due to the
decay of $e^{8i\lambda ^{3}t}$ on $\Gamma $ and one concludes from $\left( %
\ref{L1 norm}\right) $ that 
\begin{equation}
\left\Vert \Delta \mathbb{G}_{n,x,t}\right\Vert _{\mathfrak{S}%
_{1}}\rightarrow 0,:n\rightarrow \infty  \label{S norm}
\end{equation}%
for any real (and complex too) $x$ and $t>0$. Combining now (\ref{ineq1}), (%
\ref{ineq1'}), (\ref{S norm}) and taking into account \cite{SimonTrace}%
\begin{equation*}
\left\vert \det \left( I+A\right) -\det \left( I+B\right) \right\vert
\lesssim \left\Vert A-B\right\Vert _{\mathfrak{S}_{1}}e^{\left( 1+\left\Vert
A\right\Vert _{\mathfrak{S}_{1}}+\left\Vert B\right\Vert _{\mathfrak{S}%
_{1}}\right) }
\end{equation*}%
we conclude that 
\begin{equation*}
\det \left\{ I-\left( I+\mathbb{M}_{x,t}\right) ^{-1}\left( I+\mathbb{M}%
_{x,t}^{+}\right) \left( I+\mathbb{M}_{n,x,t}^{+}\right) ^{-1}\Delta \Omega
_{x,t}\right\} \rightarrow 0,:n\rightarrow \infty
\end{equation*}%
uniformly in $x$ on any compact interval for any $t>0$. Note that if a
sequence of real analytic functions $f_{n}\left( z\right) $ converges to
some real analytic function $f\left( z\right) $ uniformly on a compact set $%
K $ then $f_{n}^{\prime }\left( z\right) \rightarrow f^{\prime }\left(
z\right) $ also uniformly on $K$. Indeed, it immediately follows from the
Cauchy formula that for any function $f$ analytic on some closed disc $%
B_{r}\left( a\right) =\left\{ z:\left\vert z-a\right\vert \leq r\right\} $
one has 
\begin{equation*}
\left\Vert f^{\prime }\right\Vert _{L^{\infty }\left( B_{\varepsilon
r}\left( a\right) \right) }\lesssim \frac{\left\Vert f\right\Vert
_{L^{\infty }\left( B_{r}\left( a\right) \right) }}{\left( 1-\varepsilon
\right) r},\;0<\varepsilon <1.
\end{equation*}%
One can now easily conclude that the second term on the right hand side of (%
\ref{delta V}) converges to zero uniformly on compacts in $\mathbb{R}$ as $%
n\rightarrow \infty $ and the function $V\left( x,t\right) $ formally
defined by (\ref{Det}) is indeed a classical solution to the KdV problem. It
remains to demonstrate $\left( \ref{IC}\right) $. To this end one needs to
use the Fredholm expansion of the determinant in (\ref{Det}). The complete
expansion is unwieldy but \cite{RybIP09} its first term is the least
regular. Moreover, its component corresponding to $\chi _{+}V_{0}$ is
classical and the problem essentially boils down to showing that the $%
L_{+}^{2}$ norm of $\int_{\Gamma }i\lambda e^{2i\lambda x}\left(
e^{8i\lambda ^{3}t}-1\right) G\left( \lambda \right) d\lambda $ vanishes as $%
t\rightarrow 0$ in $L_{+}^{2}$. We have 
\begin{eqnarray*}
&&\left\Vert \int_{\Gamma }i\lambda e^{2i\lambda x}\left( e^{8i\lambda
^{3}t}-1\right) G\left( \lambda \right) d\lambda \right\Vert _{L_{+}^{2}} \\
&\leq &\int_{\Gamma }\left\vert \lambda \right\vert \left\Vert e^{2i\lambda
x}\right\Vert _{L_{+}^{2}\left( dx\right) }\left\vert e^{8i\lambda
^{3}t}-1\right\vert \left\vert G\left( \lambda \right) \right\vert
\left\vert d\lambda \right\vert \\
&=&\int_{\Gamma }\frac{\left\vert \lambda \right\vert }{2\sqrt{\func{Im}%
\lambda }}\left\vert e^{8i\lambda ^{3}t}-1\right\vert \left\vert G\left(
\lambda \right) \right\vert \left\vert d\lambda \right\vert \\
&\lesssim &\int_{\Gamma }\sqrt{\func{Im}\lambda }\left\vert e^{8i\lambda
^{3}t}-1\right\vert \left\vert G\left( \lambda \right) \right\vert
\left\vert d\lambda \right\vert \\
&\lesssim &\int_{\Gamma }\sqrt{\func{Im}\lambda }\left\vert e^{8i\lambda
^{3}t}-1\right\vert \left\vert R_{-}\left( \lambda \right) \right\vert
\left\vert d\lambda \right\vert .
\end{eqnarray*}%
To apply the Lebesgue dominated convergence theorem one merely needs to
verify that $\left( \func{Im}\lambda \right) ^{-1/2}R_{-}\left( \lambda
\right) \in L^{1}\left( \Gamma \right) $ or, due to (\ref{R}), that $\left( 
\func{Im}\lambda \right) ^{-1/2}\left( i\lambda -m_{-}\left( \lambda
^{2}\right) \right) \in L^{1}\left( \Gamma \right) $. But \cite{CarLac90}%
\begin{equation*}
\left\vert i\lambda -m_{-}\left( \lambda ^{2}\right) \right\vert \lesssim
\left\vert \int_{\mathbb{R}_{-}}e^{2i\lambda \left\vert x\right\vert
}V_{0,-}\left( x\right) dx\right\vert +\frac{1}{\left\vert \mathbb{\lambda }%
\right\vert ^{2}}\left\vert \int_{\mathbb{R}_{-}}e^{-2\func{Im}\lambda
\left\vert x\right\vert }V_{0,-}\left( x\right) ^{2}dx\right\vert
\end{equation*}%
and hence, due to (\ref{decay at -infty}), 
\begin{equation}
\left( \func{Im}\lambda \right) ^{-1/2}\int_{0}^{\infty }e^{2i\lambda
x}V_{0,-}\left( -x\right) dx\in L^{1}\left( \Gamma \right)  \label{L^1}
\end{equation}%
is to be demonstrated. But, in virtue of condition (\ref{decay at -infty}) 
\begin{eqnarray}
\int_{1}^{\infty }e^{-\func{Im}\lambda x}\left\vert V_{0,-}\left( -x\right)
\right\vert dx &=&e^{-\func{Im}\lambda }\int_{0}^{\infty }e^{-\func{Im}%
\lambda x}\left\vert V_{0,-}\left( -x-1\right) \right\vert dx  \notag \\
&=&O\left( e^{-\func{Im}\lambda }\right) ,\ \ \ \ \func{Im}\lambda
\rightarrow \infty ,  \label{1,infty}
\end{eqnarray}%
and \cite{Ry2001} 
\begin{equation}
\int_{0}^{1}e^{-2\func{Im}\lambda x}V_{0,-}\left( -x\right) dx=\frac{%
V_{0,-}\left( -0\right) }{2\func{Im}\lambda }+o\left( \frac{1}{\func{Im}%
\lambda }\right) ,\ \ \ \ \func{Im}\lambda \rightarrow \infty ,  \label{0,1}
\end{equation}%
provided that $x=0$ is a Lebesgue point of $V_{0}\left( x\right) $. Since $%
V_{0}\left( x\right) $ is locally integrable, almost every $x$ is a point of
Lebesgue continuity and thus, without loss of generality, we may assume that 
$a=0$ is Lebesgue. Putting (\ref{0,1}) and (\ref{1,infty}) together implies (%
\ref{L^1}) and thus (\ref{IC}) is proven for $a=0$. Simple shifting
arguments extend (\ref{IC}) to any real $a$.
\end{proof}

Theorem \ref{MainThm} extends our main result in \cite{Ryb10} where somewhat
weaker solutions are considered under much stronger conditions on initial
data. In addition no explicit formula (\ref{Det}) is given there either.
Theorem \ref{MainThm} considerably extends also \cite{Tarama04} where
similar results (save (\ref{Det})) are obtained. In \cite{Tarama04} $\chi
_{-}V_{0}$ is assumed to be Faddeev and from $L_{-}^{2}$ which is of course
much stronger than (\ref{decay at -infty}). The approach of \cite{Tarama04}
is also based upon inverse scattering but the analysis is conducted within
the classical Marchenko theory. We however crucially used \cite{Tarama04} to
combine the treatment of the $\mathbb{R}_{-}$ from \cite{Ryb10} with the one
given in \cite{Tarama04} for $\mathbb{R}_{+}$.


\begin{remark}
Note that Conditions $\left( \ref{decay at -infty}\right) $ and (\ref{decay
at +fty}) alone accept those $V_{0}$'s for which $\left\vert V_{0}\left(
x\right) \right\vert \rightarrow \infty $ exponentially fast as $%
x\rightarrow -\infty $ but exhibit a subexponential decay as $x\rightarrow
\infty $. Condition $\inf \limfunc{Spec}\left( -\partial
_{x}^{2}+V_{0}\right) \geq -h_{0}>-\infty $, meaning that $V_{0}\left(
x\right) $ is essentially bounded from below, however does not allow $%
V_{0}\left( x\right) \rightarrow -\infty $ as $x\rightarrow -\infty $. The
condition that $\limfunc{Spec}_{ac}\left( -\partial _{x}^{2}+V_{0}\right) $
has a non-empty component of multiplicity two, on the other hand, does allow 
$V_{0}\left( x\right) $ to go to $-\infty $ as $x\rightarrow -\infty $ but
not to $\infty $. This condition assumes certain pattern of behavior of $%
V_{0}$ when $x\rightarrow -\infty $ but it is hard to express in terms of $%
V_{0}$ alone. Loosely speaking, for an operator with a.c. spectrum, it is
possible to approximately predict future values of the potential, with
arbitrarily high accuracy, based on information about past values (so-called
Oracle Theorems). We refer the reader to \cite{Remling09} and the literature
cited therein. On the other hand, there are many explicit examples of
potentials (including those decaying like $\left\vert x\right\vert
^{-a},\alpha <1/2$) for which the spectrum is purely singular (see, e.g. 
\cite{Pearson78}, \cite{KoU88}).
\end{remark}


\begin{remark}
If we assume that $V_{0}$ is short range then the right reflection
coefficient is also well-defined by $\left( \ref{R1}\right) $. The formula $%
\left( \ref{R1}\right) $ says that $R_{+}\left( \lambda \right) $ admits an
analytic continuation into $\mathbb{C}_{-}$ but does not imply analyticity
in $\mathbb{C}_{+}$.
\end{remark}

\begin{remark}
\label{+ infty}Theorem \ref{MainThm} holds with an obvious change in (\ref%
{Det}) if we replace \ref{decay at +fty}) with $\chi _{+}V_{0}-c\in
L_{+}^{2}\left( e^{\delta _{+}x^{1/2}}dx\right) $ with some real $c$.
Indeed, performing a simple Galilean transform, one has 
\begin{equation*}
V\left( x,t\right) =c+W\left( x,t\right)
\end{equation*}%
where $W$ solve the KdV equation with initial data $W_{0}$ subject to the
conditions of Theorem \ref{MainThm}.
\end{remark}

\begin{remark}
The decay condition (\ref{decay at +fty}) in Theorem \ref{MainThm} can
likely be relaxed to read $\chi _{+}V_{0}\in L^{1}\left( \left\langle
x\right\rangle dx\right) $ (or at least to $L^{1}\left( \left\langle
x\right\rangle ^{2}dx\right) $) but $V\left( x,t\right) $ will no longer be
a real analytic function in $x$ for any $t>0$. Further relaxation of the
decay condition at $+\infty $ runs into serious problems which will be
discussed in the next section.
\end{remark}

\section{Discussions, corollaries, and open problems}

\subsection{Hirota $\protect\tau $-function}

The explicit formula (\ref{Det}) is by no means new. In the context of
integrable systems the formula 
\begin{equation}
V\left( x,t\right) =-2\partial _{x}^{2}\log f\left( x,t\right) 
\label{Hirota}
\end{equation}%
seems to have appeared first in the physical literature \cite{Hirota71} back
in early 70s and has become one of the main ingredient of soliton theory.
(or perhaps even earlier in the connection with $N$ soliton solutions.) The
nature of $f\left( x,t\right) $ varies depending on the context where it
appears in. It is used, e.g., as a substitution to transform the KdV
equation into the so-called bilinear KdV equation (Hirota's $\tau $%
-transform). Every known explicit solution can be written in the form (\ref%
{Hirota}) where $f\left( x,t\right) $ is a certain Wronskian (see, e.g. \cite%
{Ma05}) or finite determinant \cite{AktMee06}. It is also well known in
certain contexts similar to ours (see, e.g. \cite{ErcMcKean90}, \cite%
{Venak86}) and is referred to as determinant, Bargman, Dyson, etc. We
however call it Hirota's $\tau $-representation. The literature on such
formulas is enormously broad but in our generality (\ref{Det}) appears to be
new. We however were unable to prove that $\mathbb{M}_{x,t}$ is actually
trace class for any $t>0$ and the determinant in (\ref{Det}) can therefore
be understood in the classical Fredholm sense. We could not find a rigorous
proof of this fact in the literature even for short range initial profiles.
The problem is equivalent to the question what conditions on (short range) $%
V_{0}$ guarantee that the integral Hankel operator $\mathbb{H}$ on $L_{+}^{2}
$ defined by%
\begin{equation*}
\left( \mathbb{H}f\right) \left( x\right) =\int_{\mathbb{R}_{+}}H\left(
x+y\right) f\left( y\right) dy,\ \ \ x\in \mathbb{R}_{+}.
\end{equation*}%
were 
\begin{equation}
H\left( x\right) :=\int e^{i\lambda x}e^{i\lambda z+i\lambda ^{3}t}R\left(
\lambda \right) d\lambda ,  \label{M(x)}
\end{equation}%
is a trace class operator for any real $z$ and $t>0$. Peller found \cite%
{Peller2003} a complete characterization of all trace class Hankel operators
in terms of a very subtle smoothness of $H$ (membership of its Fourier
transform in the so-called analytic Besov class). However this
characterization cannot be easily expressed in term of $V_{0}$.

\subsection{Asymptotic solutions}

The formula (\ref{Det}) (and Theorem \ref{MainThm} as a whole) is
particularly convenient to study solutions of Cauchy problems for the KdV
equation. It has however been primarily used in the context of short range
reflectionless initial profiles where the Marchenko operator becomes finite
rank and the Fredholm determinant turns into a linear algebra object.
Historically, this case was studied first and the phenomenon of nonlinear
superposition and interaction of solitons was understood this way. This
analysis was then extended in \cite{Gesztesy_Duke92} to the closure of
reflectionless potentials. In the form closest to ours (\ref{Det}) seems to
have appeared in \cite{Venak86}. Namely, under the assumption that $%
V_{0}\rightarrow -c^{2}$ as $x\rightarrow -\infty $ and $V_{0}\rightarrow 0$
as $x\rightarrow \infty $ sufficiently fast it was used in \cite{Venak86} to
derive and analyze a long time asymptotic of $V\left( x,t\right) $ producing
short-cuts to most of previous results of \cite{Hruslov76}. The treatment
was justified under an extra assumption that $R=0$ which however holds only
asymptotically as $t\rightarrow \infty $. In \cite{KK94} this condition was
not imposed and a rigorous justification was done assuming only that $%
\partial _{x}V_{0}$ is from the Schwartz class. We emphasize that these
works are concerned with asymptotic solutions and WP issues are not treated
there. This is, in fact, the main focus of the present paper. We however
believe that Theorem \ref{MainThm} should be particularly useful to derive
various asymptotics generalizing and improving previously known results%
\footnote{%
Roughly speaking, a typical such result says that $V\left( x,t\right) $
asymptotically splits as $t\rightarrow \infty $ into an infinite train of
solitons of height $-2c^{2}$.} in the conditions it is proven under. It
could be achieved by a suitable approximation of the operator $\mathbb{M}%
_{x.t}$ in (\ref{Det}) which can be done in a number of different ways. We
hope to return to this important issue elsewhere.

\subsection{Conditions on a.c. spectrum}

The condition that $\func{Spec}\left( -\partial _{x}^{2}+V_{0}\right) $ is
bounded from below guarantees that no blow-up solitons will instantaneously
emerge from $-\infty $. It is however may be possible to construct an
initial condition $V_{0}$ such that the negative spectrum of $-\partial
_{x}^{2}+V_{0}$ is very spares but not bounded from below so that Theorem %
\ref{MainThm} would still hold. The physical content of such a situation
would be dubious though.

The condition that $\limfunc{Spec}\left( -\partial _{x}^{2}+V_{0}\right) $
has a non-trivial a.c. component of multiplicity two does not appear to be
physically motivated. While satisfied in all previously studied cases
mentioned in Remark \ref{Known cases} this condition is hard to verify in
general and only vague statements can be made as of today explaining what it
means for $V_{0}$ to satisfy such condition. Theorem \ref{Structure}
guarantees that the determinant in (\ref{Det}) does not vanish on the real
line and hence $V\left( x,t\right) $ has no real poles. The latter means no
blow-up solutions develop over finite time. However as mentioned above one
can explicitly construct initial profiles (e.g. Pearson blocks running to $%
-\infty $ and zero on the right) for which the spectrum is positive and its
a.c. component fills $\mathbb{R}_{+}$ but has uniform multiplicity one.
Lemma \ref{Invert} does not apply as $\left\vert R\left( \lambda \right)
\right\vert =1$ a.e. on the real line. Note that this problem has embedded
dense singular spectrum. Yet another example comes from a white noise type
random initial profile supported on the left half line. We however
conjecture that the assertion of Lemma \ref{Invert} would hold. In a
somewhat simplified case (absence of negative spectrum), the problem
essentially boils down to the following question. If the operator 
\begin{equation*}
I+\chi _{+}\mathcal{F}\phi \mathcal{F}:L_{+}^{2}\rightarrow L_{+}^{2}
\end{equation*}%
is boundedly invertible for any unimodular function $\phi $ of the form 
\begin{equation*}
\phi \left( \lambda \right) =e^{i\left( \lambda ^{3}t+c\lambda \right)
}I\left( \lambda \right)
\end{equation*}
where $t\geq 0,c\in \mathbb{R\ }$are constants and $I\left( \lambda \right) $
is an inner function of the upper half plane (i.e. $I\in H_{+}^{\infty }$
and such that $\left\vert I\left( \lambda \right) \right\vert <1$ in $%
\mathbb{C}_{+}$ and $\left\vert I\left( \lambda \right) \right\vert =1$ a.e.
on the real line)? This problem can be viewed as a Riemann-Hilbert problem
or a problem about invertibility of the Toeplitz operator with the
(unimodular) symbol $\phi $. Since $e^{i\lambda ^{3}t}$ is not a Nevanlinna
function, the powerful machinery of Riemann-Hilbert problem or
Hankel/Toeplitz operators does not readily apply. Moreover, the Bourgain
factorization theorem for a unimodular function \cite{Bourgain86} suggests
that the answer may actually be negative\footnote{%
We thank Donald Marshall for bringing out attention to this paper.} meaning
that, loosely speaking, positon solutions (with double pole singularities)
will emerge from noise. The latter would mean that either shallow water
rouge waves are likely a stochastic phenomenon or the KdV equation does not
model the situation good enough in this case. We however cautiously
conjecture that the answer is affirmative and moreover the conclusion of
Theorem \ref{MainThm} holds without the hard-to-verify condition on the a.c.
spectrum \ An improved theorem \ref{MainThm} would give solutions with
simple a.c. and dense singular spectra (produced, e.g., by Pearson blocks)
which would of course be a new type of KdV solutions.

\subsection{Uniqueness results}

The following statement is a direct consequence of the analyticity of $%
V\left( x,t\right) $ for $t>0$.

\begin{corollary}
\label{Corollary'}Under conditions of Theorem \ref{MainThm} the solution $%
V\left( x,t\right) $ can not vanish on a set of positive Lebesgue measure
for any $t>0$ unless $V_{0}$ is identically zero.
\end{corollary}

This extends a result due to Zhang \cite{Zhang92} stating that if $V\left(
x,t\right) $ is a solution of the KdV equation then it cannot have compact
support at two different moments unless it vanishes identically. Corollary %
\ref{Corollary'} quickly recovers and improves yet another important result
of \cite{Zhang92} saying that, assuming $V_{0}\in L^{1}\left( \left\langle
x\right\rangle dx\right) $ and $\partial _{x}V_{0}\in L^{1}$, once $\limfunc{%
Supp}V_{0}\left( x,t_{0}\right) ,\limfunc{Supp}V\left( x,t_{1}\right)
\subset \left( -\infty ,a\right) $ for $t_{0}<t_{1}$ then $V\left(
x,t\right) =0$. The techniques of \cite{Zhang92} use the Marchenko theory
and some Hardy space arguments. Zhang's approach has been simplified and
generalized in recent \cite{Teschl3}. The main results of \cite{Zhang92}
were extended in \cite{Escau07}. Under certain regularity and decay
conditions, the main result of \cite{Escau07} says that if at two distinct
times two solutions differ by super exponentially decaying functions then
the solutions must coincide. We believe that this result can be improved by
our techniques and will return to this elsewhere. Note that \cite{Escau07}
treats a more general class of KdV type equations for which the IST need not
apply.

\subsection{Analyticity}

It is proven in \cite{Deift79} that if $V_{0}\left( x\right) $ is analytic
in the strip $\left\vert \func{Im}z\right\vert <a$ and has Schwartz decay
there, then $V\left( x,t\right) $ is meromorphic in a strip with at most $N$
poles (where $N$ is the number of bound states of $-\partial
_{x}^{2}+V_{0}\left( x\right) $) off the real line. Note that in this
situation the reflection coefficient $R$ necessarily exhibits an exponential
decay on the real line. This of course need not occur in our case (even if
we assume that $V_{0}$ is supported on $\mathbb{R}_{-}$) and our real
meromorphic solution has, in general, infinitely many poles for any $t>0$ in
any strip around the real line (but not on it) accumulating only to
infinity. These poles analytically depend on $t$ and hence may appear or
disappear only on the boundary of analyticity of $V\left( x,t\right) $. In
the case of $V_{0}$ supported on $\mathbb{R}_{-}$, $V\left( x,t\right) $ is
meromorphic on the whole complex plane meaning that the poles can only move.
In fact, the importance of pole dynamics was recognized by Kruskal\ \cite%
{Kruskal74} back in the early 70s for pure soliton solutions and has been
actively studied since then (see also \cite{Air77}, \cite{Segur2000}, \cite%
{GWeikard06}, \cite{Bona2009} to mention just four). Now existence of
meromorphic solutions is referred to as the Painlev\'{e} property (see, e.g. 
\cite{McLeodOlver83} and \cite{Its2003}) \ Theorem \ref{MainThm} says the
KdV with the initial profile $V_{0}$ supported on $\mathbb{R}_{-}$ has the
Painlev\'{e} property. However we cannot say much about the structure of our
poles and their dynamics and we don't know how one can spot the future poles
from the (non-analytic) initial profile and how they actually appear. In
particular, we have no norm estimates or any other quantitative way to
express the `size' for this function in terms of initial data. Such
estimates could likely follow from conservation laws which typically
accompany the IST. For initial data which tends to a constant at $-\infty $
deriving such formulas should not be a problem. We are not sure if
conservation laws have been derived for general step like initial data. We
do not know if our real meromorphic solution is meromorphic on the whole
complex plain. The answer depends the same question in the setting of \cite%
{Tarama04}. The techniques used in \cite{Tarama04} yield only local
analyticity around the real line (real analyticity) and cannot be easily
adjusted to obtain global analyticity. We however believe that it is
achievable but by different methods. In \cite{McLeodOlver83} it is claimed
without a proof that $V\left( x,t\right) $ is meromorphic for any $t>0$
under the only assumption that $V_{0}\in L^{1}\left( \left\langle
x\right\rangle ^{2}dx\right) $. This would likely be the case if the Fourier
transform of $V_{0}$ admitted an analytic continuation into the lower half
plane. Unfortunately, an arbitrary function from $L^{1}\left( \left\langle
x\right\rangle ^{2}dx\right) $ doesn't have this property\footnote{%
We owe Eugine Korotyev for this comment.}.

We have not looked into the question if there are nontrivial initial
profiles for which $V\left( x,t\right) $ is an entire function of $x$ for
any $t>0$. The answer is likely affirmative as the linear part of the KdV
equation (Airy's equation) has a very strong smoothing property. It is not
hard to show that this equation instantaneously evolves a multisoliton
initial profile, which is a meromorphic functions on the whole complex with
infinitely many double poles, into an entire function. On the other hand the
non-linear part of the KdV (Hopf's equation) tends to break analyticity.

\subsection{Slower or no decay at $+\infty $}

Due to the directional anisotropy of the KdV equation (related to the
different decay behavior of the Airy functions) we cannot switch in Theorem %
\ref{MainThm} $V_{0}\left( x\right) $ to $V_{0}\left( -x\right) $. Theorem %
\ref{MainThm} says that initial data on $\left( -\infty ,a\right) $
instantaneously evolves into a meromorphic function (the so-called gain of
regularity) for any finite $a$ and causes no problem to WP. Contrary to
this, even small data on $\left( a,\infty \right) $ may reduce smoothness of
a solution (loss of regularly) and even lead to a blow-up solution. This
phenomenon can be drastically seen in the case of a delta initial profile
which evolves into a single soliton plus an Airy type function. Due to the
invariance of the KdV equation with respect to $\left( x,t\right)
\rightarrow \left( -x,-t\right) $, one has an example of a smooth
(meromorphic) initial profile rapidly decaying at $-\infty $ and exhibiting
slow oscillatory decay at $+\infty $ which turns at a certain instant of
time into a delta function (focusing effect). This delta function (which
could formally be considered as a compactly supported initial data) then
instantaneously evolves into a meromorphic function which stays meromorphic
for all times. Note that since the delta function is in $H^{-3/4}$ the
general theory \cite{Tao06} guarantees that this problem is globally WP.
This makes one believe that IST techniques should work in the general
setting of $H^{-3/4}$ data\footnote{%
Better yet, the recent paper \cite{KapPerryTopalov2005} says that the KdV
equation is WP for a certain subset of $H^{-1}$ functions. More
specifically, the image of $L^{2}$ under the Miura tranform.}. However, no
rigorous IST method, to the best to our knowledge, is developed in this case
in full generality. Comparing to the classical IST, there are new
circumstances like infinite negative spectrum and a rich singular embedded
positive spectrum in this setting. The importance of this problem was
recognized back in the 80s and the first essential progress was made in \cite%
{KheNov84} where the IST was extended to certain Wigner-von Neumann initial
profiles producing finitely many simple real zeros of the transmission
coefficient. This topic has been continued in numerous publications (see,
e.g. \cite{Matveev02}, \cite{Kovalyov05} and the literature cited therein).
The main issue here is the appearance of the so-called positon solutions.
These solutions are in $H^{-2}$ but exhibit Wigner-von Neumann decay at
infinity. The Schrodinger operator with a $H^{-2}$ potential can be defined
in a few nonequivalent ways resulting in, at least, two nonequivalent IST
formalisms (see \cite{Kurasov99} and the relevant discussions therein).

The situation is much better in the context of initial data with periodic
type behavior at $+\infty $ (see recent \cite{Teschl2} and a very detailed
account of the literature therein). As discussed above, suitable analogs of
the IST are available in this setting which allow to push WP results to the
periodic $H^{-1}$ space. However the KdV flow preserves the smoothness of a
periodic initial profile confirming that the data on $\left( a,\infty
\right) $ works against the gain of regularity effect\footnote{%
We don't know how the interference between the data $\left( -\infty
,a\right) $ and $\left( a,\infty \right) $ affects regularity either.}.

We however cautiously conjecture that the Hirota $\tau $-representation (\ref%
{Det}) holds in one form or another for any well-posed KdV problem. This
would mean that KdV equation is completely (and explicitly) integrable for a
large scope of problems. The arguments used in the proof of Theorem \ref%
{MainThm} no longer apply to profiles with no decay at $+\infty $\footnote{%
Due to Remark \ref{+ infty}, initial data approaching a constant at $+\infty 
$ is not of interest to us.}. It would be extremely interesting to
understand this situation even in the case of the data supported on $\left(
0,\infty \right) $.

\section*{Acknowledgement}

The author is thankful to Eugene Korotyaev for valuable discussions related
to the periodic KdV and Odile Bastille for careful proof reading of the
manuscript.

\end{document}